\documentclass[11pt]{article} 
\usepackage[a4paper, total={6.1in, 8.7in}]{geometry}

\usepackage{layout}
\usepackage{amsmath,amssymb,amsfonts, amsthm}
\usepackage{todonotes}
\usepackage{mathtools}
\usepackage{enumitem}
\usepackage{authblk}
\newtheorem{theorem}{Theorem}[section]
\newtheorem{proposition}[theorem]{Proposition}
\newtheorem{definition}[theorem]{Definition}
\newtheorem{claim}[theorem]{Claim}
\newtheorem{lemma}[theorem]{Lemma}
\newtheorem*{lemma*}{Lemma}

\newtheorem{conjecture}[theorem]{Conjecture}

\newtheorem{corollary}[theorem]{Corollary}

\newtheorem{observation}[theorem]{Observation}
\newtheorem*{observation*}{Observation}

\newif\ifdraft
\drafttrue

\renewcommand\epsilon\varepsilon
\newcommand{\eps}{\varepsilon}
\newcommand{\eq}[1]{\begin{align*}#1\end{align*}}
\newcommand{\eql}[2]{\begin{align}\label{#2}#1\end{align}}

\newcommand{\F}[2]{{\frac{#1}{#2}}}
\newcommand{\R}[1]{{\frac{1}{#1}}}
\newcommand{\B}[1]{\left( {#1} \right)}
\newcommand{\BF}[2]{\B{\F{#1}{#2}}}

\newcommand{\RL}{PF}
\newcommand{\reach}{t}
\newcommand{\period}{\delta}
\newcommand{\algzplus}{\texttt{A}_{walk}}

\newcommand{\algy}{\texttt{A}_{loop}}
\newcommand{\alg}{\texttt{A}}
\newcommand{\algsep}{\texttt{A}_{sep}}
\newcommand{\algcontract}{\texttt{A}_{query}}

\newcommand{\algregular}{\texttt{A}_{2-\text{layers}}}

\newcommand{\algtwosepcost}{
\sqrt{\Delta} \log n \cdot \log \log n}
\newcommand{\I}{\item}

\newcommand{\kside}{\F 1 6}
\newcommand{\kup}{\F 1 {12}}

\newcommand\RR{\mathbb{R}}

\newcommand{\constq}{(1-\varepsilon)}
\newcommand{\thresh}{m}
\newcommand{\smallh}{h_{2}}
\newcommand{\bigh}{h_{1}}

\newcommand{\newsource}{\sigma'}
\newcommand{\newtarget}{\tau_u}
\newcommand{\faulty}{faulty\,}

\def\NL{{N_{{layer}}}}

\newcommand{\condition}{($\star$)~}
\def\t2{(\log n)^{10}}

\def\queries{{\mathcal{Q}}}

\newcommand\dist{d}

\newcommand\adv{{\tt{adv}}}
\newcommand\advTo[1]{\overrightarrow{\adv}(#1)}
\newcommand\advAway[1]{\overleftarrow{\adv}(#1)}

\newcommand\E{\mathbb{E}}

\newcommand\Lvl{{\cal{L}}}

\newcommand\PP{\mathbb{P}}
\newcommand\EE{\mathbb{E}}
\newcommand\QE[1][x]{E_{#1}}

\newcommand{\listen}{\lambda}
\newcommand{\good}{\textup{\texttt{Good}}~}

\newcommand{\source}{\sigma}
\newcommand{\target}{\tau}
\newcommand{\treasure}{\target}

\newcommand{\diam}{D}
\newcommand{\depth}{D}
\newcommand{\fast}{{\texttt{A}_\textit{fast}}~}

\newcommand{\intermediate}{{\texttt{A}_\textit{mid}}~}

\newcommand{\excellent}{\textup{\texttt{Excellent}}~}
\newcommand{\excellentc}{\textup{\texttt{Excellent$^c$}}~}

\newcommand{\pathoo}[2]{\langle#1 ,#2\rangle}
\newcommand{\pathoc}[2]{\langle#1 ,#2]}
\newcommand{\pathco}[2]{[#1 ,#2\rangle}
\newcommand{\pathcc}[2]{[#1 ,#2]}
\newcommand{\path}[2]{[#1,#2]}
\newcommand{\leaf}[1][]{{\tau_{#1}}}

\newcommand{\set}[1]{\left\{ #1 \right\}}
\newcommand{\midline}[2]{{#1} \, \middle| \, {#2}}
\newcommand{\prob}[1]{\mathbb{P}\left( #1 \right)}
\newcommand{\sexpct}[2]{\mathbb{E}_{{#1}}\left[#2\right]}

\newcommand{\expct}[1]{\sexpct{}{#1}}

\newcommand{\cprob}[2]{\prob{\midline{#1}{#2}}}
\newcommand{\cexpct}[2]{\expct{\midline{#1}{#2}}}

\newcommand\score{\texttt{score}}
\newcommand{\local}[1][h]{{\textup{\texttt{local}}_{#1}}}

\newcommand{\Ceil}[1]{{\left\lceil {#1} \right\rceil}}

\newcommand\leaves{\mathcal{L}}
\newcommand{\badsides}[2][S]{M^{#1}_{\text{sides}}(#2)}
\newcommand{\badup}[2][S]{M^{#1}_{\text{up}}(#2)}

\newcommand{\repeatable}[5]{
\begin{#1}{\label{#4}}
{#5}
\end{#1}
\theoremstyle{#3}
\newtheorem*{temp#4}{#2 \ref{#4} \textit{(restated)}}
\global\expandafter\def\csname #4\endcsname{
\begin{temp#4}
{#5}
\end{temp#4}}
}

\def\ThmStyle{plain}

\newcommand{\rLemma}[2]{\repeatable{lemma}{Lemma}{\ThmStyle}{#1}{#2}}

\usepackage[lined]{algorithm2e}
\RestyleAlgo{boxruled}
\LinesNumbered

\newcommand{\compleaves}{B}

\newcommand{\specT}{\textbf{H}}
\newcommand{\subT}{H}

\newcommand{\ldepth}{h}
\newcommand{\smallfun}{C}

\usepackage{bbm}
\usepackage{xspace}

\renewcommand{\epsilon}{\varepsilon}

\renewcommand{\Pr}{P}

\renewcommand{\epsilon}{\varepsilon}

\newcommand\bigO{\mathcal{O}}

\usepackage{float}
\floatstyle{boxed}
\newfloat{pseudocode}{thb}{pseudo}
\floatname{pseudocode}{Algorithm}

\usepackage{hyperref}
\usepackage{cleveref}

\hypersetup{
    colorlinks = false,
    linkbordercolor = {white},
    citebordercolor = {blue}
}

\usepackage{mwe}

\date{}

\title{
Searching Trees with Permanently Noisy Advice:\\
Walking and Query Algorithms
\footnote{This work has received funding from the European Research Council (ERC) under the European Union's Horizon 2020 research and innovation programme (grant  agreement No 648032). This work was also supported in part by the Israel Science Foundation (grant No. 1388/16).  A preliminary version of this paper  appeared in ESA 2018. The current journal version contains many additional results.
}
}
\author[1]{Lucas Boczkowski}
\author[2]{Uriel Feige}
\author[1]{Amos Korman}
\author[3]{Yoav Rodeh}

\affil[1]{CNRS, IRIF, Univ Paris Diderot, Paris, France.}
\affil[2]{The Weizmann Institute of Science, Rehovot, Israel.}
\affil[3]{Ort-Braude College, Karmiel, Israel.}

\begin{document}
\maketitle
\begin{abstract}
We consider a search problem on trees in which the goal is to find an adversarially placed treasure, while relying on local, partial information. Specifically, each node in the tree holds a  pointer to one of its neighbors, termed \emph{advice}. A node is faulty with probability $q$. The advice at a non-faulty node points to the neighbor that is closer to the treasure, and the advice at a faulty node points to a uniformly random neighbor. Crucially, the advice is {\em permanent}, in the sense that querying the same node again would yield the same answer.

Let $\Delta$ denote the maximal degree. Roughly speaking, when considering the expected number of {\em moves}, i.e., edge traversals, we show that a phase transition occurs when the {\em noise parameter} $q$ is about $1/\sqrt{\Delta}$. Below the threshold, there  exists an algorithm with expected move complexity $\bigO(D\sqrt{\Delta})$, where $D$ is the depth of the treasure, whereas above the threshold, every search algorithm has expected number of moves which is both exponential in $D$ and polynomial in the number of nodes~$n$.

In contrast, if we require to find the treasure with probability at least $1-\delta$, then
for every fixed $\varepsilon > 0$, if $q<1/\Delta^{\varepsilon}$ then there exists a search strategy that with probability $1-\delta$ finds the treasure using $(\delta^{-1}D)^{O(\frac 1 \varepsilon)}$ moves. Moreover, we show that $(\delta^{-1}D)^{\Omega(\frac 1 \varepsilon)}$ moves are necessary.

Besides the number of moves, we also study the number of advice {\em queries} required to find the treasure. Roughly speaking, for this complexity, we show similar threshold results to those previously stated, where the parameter $D$ is replaced by $\log n$.

\end{abstract}
\newpage
\tableofcontents

\section{Introduction}
This paper considers a search problem on trees, in which the goal is to find a treasure that is placed at one of the nodes by an adversary. Each node of the tree holds information, called {\em advice}, regarding which of its neighbors is closer to the treasure, and the algorithm may query the advice at some nodes in order to accelerate the search. In this paper we study both the number of walking steps and the number of advice queries needed in order to find the treasure.

Searching with advice on trees is an extension of binary search to tree topologies.
This type of extension has been the focus of numerous works \cite{Spirakis,BinarySearchFramework,Kempe16,Feige94, Onak08, Onak06}, some including noise or errors in the advice. The problem may also be viewed as searching a poset \cite{Onak06, Onak08}, instead of a completely ordered set as in typical binary search. Some authors also motivate the problem using the notion of ``bug detection", where the tree models dependencies between programs \cite{Onak08}.
When the searcher is restricted to walk on the edges of the underlying graph, it is possible to view the problem as a routing problem with unreliable local information \cite{eLife, Hanusse08, Kos10}. An interesting application was also given in \cite{NIPS}, in the context of interactive learning. 

The crucial feature of our model, that distinguishes it from most existing literature on search with noisy advice, is the permanent nature of the faults. Given the tree and the location of the treasure, there is a sampling procedure (which may be partly controlled by an adversary) that determines the advice at every node of the tree. Depending on the outcome of the sampling procedure, the advice at a node may either be correct or faulty (we also refer to the latter case as noise). The advice is
{\em permanent} -- it does not change after the sampling stage. Every query to a given node yields the same advice -- there is no re-sampling of advice. The difference between permanent noise and re-sampled one (as in e.g., \cite{bayes08,Kempe16,Feige94,Karp07}) is dramatic, since the re-sampled advice model allows algorithms to boost the confidence in any given piece of advice by repeatedly querying the same advice. Permanent noise was
considered in~\cite{Braverman08} for the task of sorting, but this task is very different than the search task considered in our paper (in particular, no algorithm can find the true sorted order when noise is permanent). Searching with permanent faulty nodes has also been studied in a number of works \cite{Brodal07,memory2,Kos10,Hanusse08,Hanusse04}, but assuming that the faulty nodes are chosen by an adversary. The difference between such worst case scenarios and the probabilistic version studied here is again significant, both in terms of results and in terms of techniques (see more details in Section \ref{sec:related}).

The model of permanent faults aims to model faults that occur in the physical memory associated with the node, rather than, for example, the noise that is associated with the actual mechanism behind the query.  Interestingly, the topic of noisy permanent advice is also meaningful outside the realm of classical computer science, and was shown to be relevant in the context of ant navigation \cite{eLife}. The authors therein conducted experiments in which a group of ants carry a large load of food aiming to transport it to their nest, while basing their navigation on unreliable advice  given by pheromones that are laid on the terrain. Indeed, although the directions proposed by pheromones typically lead to the nest,  trajectories as experienced by   small ants may be inaccessible to the load, and hence directional cues left by  ants sometimes lead the load towards dead-ends.

The current paper introduces the algorithmic study of search with permanent probabilistically noisy advice. Similarly to many other works on search  we focus on trees, which is a very important topological structure in computer science. Extending our work to general graphs seems technically challenging and remains for future work, see Section \ref{sec:open}.



\subsection{The Noisy Advice Model}
We start with some notation. Additional notation is introduced in Section \ref{sec:notations}.
We present the model for trees, but we remark that the definitions can be extended to general graphs (see also Section \ref{sec:open}).
Let $T$ be an $n$-node tree rooted at some arbitrary node $\source$.
 We consider an agent that is initially located at the root $\source$ of $T$,
aiming to find a node $\treasure$, called the {\em treasure}, which is chosen by an adversary.
That is, the goal of the agent is to be located at $\treasure$, and once it is there, the algorithm terminates.

The {\em distance} $\dist(u,v)$ is the number of edges on the path between $u$ and $v$.
The {\em depth} of a node $u$ is $d(u)=\dist(\source,u)$. Let $d=d(\treasure)$ denote the depth of $\treasure$, and let
the depth $\diam$ of the tree be the maximum depth of a node.
Finally, let $\Delta_u$ denote the degree of node $u$ and let $\Delta$ denote the maximum degree in the tree. For an integer  $\Delta\geq 2$, a {\em complete  $\Delta$-ary} tree is a tree such that every internal node has degree precisely $\Delta$. 

Each node $u\neq \treasure$ is assumed to be provided with an {\em advice}, termed $\adv(u)$, which provides information regarding the direction of the treasure. Specifically, $\adv(u)$ is a pointer to one of $u$'s neighbors. It is called {\em correct} if
the pointed neighbor is one step closer to the treasure than $u$ is.
Each node $u\neq \treasure$ is {\em faulty} with probability $q$ (the meaning of being faulty will soon be explained), independently of other nodes. Otherwise, $u$ is considered {\em sound}, in which case its advice is correct. We call $q$ the {\em noise parameter}. Unless otherwise stated, this parameter is the same across all nodes, but in some occasions, we also allow it to vary across nodes. In that case $q$ is defined as $\max_u(q_u)$.

\paragraph{Random and semi-adversarial variants.} We consider two models for faulty nodes. The main model assumes that the advice at a faulty node points to one of its neighbors chosen uniformly at random (and so possibly pointing at the correct one). We also consider an adversarial variant, called the {\em semi-adversarial model}, where this neighbor is chosen by an adversary. The adversary may either be {\em oblivious} or {\em adaptive}. An oblivious adversary first decides on adversarial advice for each node, afterwards each node becomes faulty independently with probability $q$, and then the true advice of faulty nodes is replaced by the respective adversarial advice. An adaptive adversary first sees the locations of all faulty nodes and only afterwards decides on the advice at the faulty nodes.
\paragraph{Move and query complexities.} The agent can move by traversing edges of the tree. At any time, the agent can query its hosting node in order to ``see'' the corresponding advice and to detect whether the treasure is present there. The search terminates when the agent queries the treasure.
In this paper two complexity measures are used. The {\em move complexity} is the number of edge traversals. The {\em query complexity}, is the number of queries made. The number of queries is always smaller than the number of moves. It is implicit when considering the query complexity that the algorithm ``knows" the topology of the tree $T$, because this topology can be learned without spending any queries.
\paragraph{Noise assumption.}
The noise parameter $q$ governs the accuracy of the environment. If $q = 0$ for all nodes, then advice is always correct. This case allows to find the treasure in $D$ moves, by simply following each encountered advice. On the other extreme, if $q=1$, then advice is essentially meaningless, and the search cannot be expected to be efficient.
An intriguing question is therefore to identify the largest value of $q$ that allows for efficient search.
\paragraph{Expectation and high probability.}
Importantly, we consider two kinds of guarantees: expectation, and high probability. In the first case, we measure the performance of the algorithm as the expected number of moves before the treasure is found.
Expectation is taken over both the randomness involved in sampling the noisy advice, and over the possible probabilistic choices made by the search algorithm.
In the second case, we consider the number of moves spent by algorithms that find the treasure with high probability (say, probability $0.9$). An upper bound on expectation can be converted into a high probability upper bound, by use of the Markov inequality. However, the converse need not hold. Indeed, as our work shows, in our setting the two kind of guarantees lead to quite different thresholds and techniques.



\paragraph{Full-information model.} For lower bound purposes, we find it instructive to also consider the following {\em full-information} model. Here the structure of the tree is known, the algorithm is given as input the advice of all nodes except for the leaves, and the treasure is at one of the leaves. The queried node can be an arbitrary leaf, and the answer reveals whether the leaf holds the treasure.

\subsection{Our Results}

We introduce the algorithmic study of search problems with probabilistic permanent faulty advice. The results all assume the underlying graph is a tree. Our results can be grouped along two axes. One is the complexity measure, which can refer to either {\em move} or {\em query}. The other is the convergence guarantee, which can refer to either {\em expectation} or {\em high probability}. We choose to split the results first according to the complexity measure. This is because, to a large extent, the query algorithms rely on good walking algorithms. The paper thus begins with the study of walking algorithms. Within this setting, we start with the average case, i.e., fast convergence in expectation.

We note however that a lower bound in the query model translates to one in the walk model. In fact, all our lower bounds are stated for queries.

\paragraph{Results in the walk model.}
Consider the noisy advice model on trees with maximum degree $\Delta$.
Roughly speaking, we show that  $1/\sqrt{\Delta}$ is the threshold for the noise parameter $q$, in order to obtain search algorithms with low expected complexities. Essentially, above the threshold, there exists trees (specifically, complete  $\Delta$-ary trees) such that for any algorithm, the expected number of moves required to  find the treasure is exponential $d$, the depth of the treasure. Conversely, below the threshold there exists an algorithm whose expected move complexity is almost linear, that is, $\bigO(d\sqrt{\Delta})$.

The proof that there is no algorithm with a small expected number of moves when the noise exceeds
$1/\sqrt{\Delta - 1}$ is rather simple. In fact, it holds even in the full information model. Intuitively,
the argument is as follows (the formal proof appears in Section~\ref{sec:lower-exp}). Consider a complete $\Delta$-ary tree of depth $D$ and assume that the treasure $\tau$ is placed at a leaf. The first observation is that the expected number of leaves having more advice point to them than to $\tau$  is a lower bound on the query complexity. The next observation is that there are more than $(\Delta-1)^D$ leaves whose distance from $\tau$  is $2D$, and for each of those leaves $u$, the probability that more advice points towards it than towards $\tau$ can be approximated by the probability that all nodes on path connecting $u$ and $\tau$ are faulty. As this latter probability is $q^{2D}$, the expected number of leaves that have more pointers leading to them is roughly $(\Delta-1)^{D} q^{2D}$. This term explodes when $q > 1/\sqrt{\Delta-1}$.

One of the main challenges we faced in the paper was show that for noise probability below $1/\sqrt{\Delta-1}$ (by a constant factor) the lower bound no longer holds, and in fact, there are extremely efficient algorithms.
Interestingly, the optimal algorithm we present is based on a Bayesian approach, which assumes the treasure location is random,
yet it works even under worst case assumptions. The challenging part in the construction was identifying the correct prior. Constructing algorithms that ensure worst-case guarantees through a Bayesian approach was done in \cite{bayes08} which studies a closely related, yet much simpler problem of search on the line. Apart from \cite{bayes08} we are unaware of other works that follow this approach.

We also analyze the oblivious semi-adversarial model, and show that the expected move complexity has a threshold also in this model, but it is much lower, around $1/\Delta$.

We then turn our attention to studying the move complexity under a given probability guarantee. We show that for every fixed $\epsilon > 0$, if $q<1/\Delta^{\varepsilon}$ then there exists a search strategy that with probability $1-\delta$ finds the treasure using $(\delta^{-1}d)^{O(\frac 1 \varepsilon)}$ moves. Moreover, we show that $(\delta^{-1}d)^{\Omega(\frac 1 \varepsilon)}$ moves are necessary. The upper bounds hold even in the adaptive semi-adversarial variant, whereas the lower bound holds even in the  purely randomized variant.

The key concept towards proving the upper bound is a notion of node \emph{fitness}.
Essentially, a node is declared fit if it has many pointers to it on the path coming from the root. This is good evidence that the node is either on the path to the treasure, or at least not too far from it. The idea is to explore the component of fit nodes to which the root, i.e., the starting point, belongs. If the component contains the treasure, the nodes on the root to treasure path, and not too many additional nodes, then the treasure is found quickly. With the appropriate formalization of fitness, efficient search can be achieved with high probability.

Unlike the lower bound on the expected number of queries, the lower bound in the high probability case 
uses the fact that obtaining the advice of nodes requires spending queries. Hence this lower bound  does not hold in the full information model.  Consider the complete $\Delta$-ary tree of depth $D$, and assume that the treasure is at a leaf. Set $q=\Delta^{-\varepsilon}$ and $h=\epsilon^{-1} \log_\Delta (D/\delta)$. Consider the length $D$ path from the root to the treasure.
On this path, with probability at least $\delta$, there exists a segment of length $h$, where all nodes are faulty. Let us denote by $H$ the subtree rooted at the highest endpoint of such a segment of $h$ consecutive faulty nodes. The algorithm needs to explore at least a constant fraction of $H$  before finding how to proceed towards the treasure. The lower bound follows as the size of the
subtree $H$ is $\Delta^h = (D/\delta)^{\varepsilon^{-1}}$.

As is always the case with non-negative random variables, the median cannot be much larger than the average, but it might be much smaller.
Our results imply that in the context of searching with noisy permanent advice, in a large regime of noise, there is an exponential gap between the median and the average. The high expectation running time is the consequence of a small fraction of the possible error patterns (the pattern of errors in the advice) for which the  search is very slow, but for almost all error patterns, the treasure is found much faster than what the high expectation suggests.

Figure~\ref{table:main-walk} summarizes the results presented in this paper regarding walking algorithms.
\begin{figure}
\begin{center}
\def\arraystretch{2}
  \begin{tabular}{c|c|c||c|c|}
    \cline{2-5}
     & \multicolumn{2}{c||}{Upper Bound } & \multicolumn{2}{c|}{Lower Bound } \\
        \cline{2-5} & Regime & Moves & Regime & Moves \\ \hline
    \multicolumn{1}{|c|}{Expectation} & $q \ll \frac 1{\sqrt{\Delta}}$ & $\bigO(d\sqrt{\Delta})$ & $q \gg \frac 1{\sqrt{\Delta}}$ & $e^{\Omega(d)}$ \\ \hline
    \multicolumn{1}{|c|}{Expectation (S.A.)} & $q \ll \frac 1{\Delta}$ & $\bigO(d)$ & $q \gg \frac 1{\Delta}$ & $e^{\Omega(d)}$ \\ \hline
    \multicolumn{1}{|c|}{High Probability} & $q < \Delta^{-\varepsilon}$& $d^{O(\varepsilon^{-1})}$ & $q > \Delta^{-\varepsilon}$ & $d^{\Omega(\varepsilon^{-1})}$ \\
    \hline
  \end{tabular}
  \caption{\label{table:main-walk} A summary of our results regarding the move complexity, in a simplified form.  The precise definition of the symbol $\ll$ will be clarified later. S.A. stands for oblivious semi-adversarial. The High Probability upper bound includes the adaptive semi-adversarial model.}
\end{center}
\end{figure}

\paragraph{Results in the query model.}
The number of queries is measured with respect to the total size of the tree $n$, rather than the distance to the treasure $d$. Roughly, we find that we can translate any result for walking algorithms, positive or negative, into one about query algorithms, replacing the depth of the treasure $d$ by $\log n$.

We build on a separator-based scheme, that would find the treasure using $O(\log n)$ queries in the absence of faults. Since the advice is not fully reliable, a mechanism is needed to catch the error in case the advice at a separator node is faulty. Probing the whole neighborhood around a separator node would be too costly. Hence, we resort to a local exploration which actually corresponds to a walking algorithm similar to the ones described earlier. The local exploration corrects potential errors at all $O(\log n)$ separator nodes on the way to the treasure, with high probability, thus leading to efficient algorithms in the high probability setting.

Local exploration may however fail due to a large quantity of errors around a separator node. A simple remedy is to settle this case by querying the whole tree.
This gives an $O(\sqrt{\Delta}\log \Delta \log^2 n)$ algorithm. To derive our best query-algorithm, we use two scales of local exploration. The higher scale is used as a fallback in case local exploration at the first scale fails.

Figure~\ref{table:main-queries} contains a summary of our results for query algorithms.
\begin{figure}
\begin{center}
\def\arraystretch{2}
  \begin{tabular}{c|c|c||c|c|}
    \cline{2-5}
     & \multicolumn{2}{c||}{Upper Bound } & \multicolumn{2}{c|}{Lower Bound } \\ \cline{2-5}
        \cline{2-5} & Regime & Queries & Regime & Queries \\ \hline
    \multicolumn{1}{|c|}{Expectation} & $q \ll \frac 1{\sqrt{\Delta}}$ & $\tilde{\bigO}(\sqrt{\Delta}\log n)$ & $q \gg \frac 1{\sqrt{\Delta}}$ & $n^{\Omega(1)}$ \\ \hline
    \multicolumn{1}{|c|}{High Probability} & $q = \Delta^{-\varepsilon}$& $(\log n)^{O(\varepsilon^{-1})}$ & $q = \Delta^{-\varepsilon}$ & $(\log n)^{\Omega(\varepsilon^{-1})}$ \\
    \hline
  \end{tabular}
  \caption{\label{table:main-queries} Query complexity results, in simplified form. The precise conditions behind the symbol $\ll$ will be clarified later.}
\end{center}
\end{figure}

\subsubsection{Results in Expectation for Walking Algorithms}
\label{sec:resultsinexp}

In Section \ref{sec:walk-exp} we present an algorithm whose expected move complexity is optimal up to a constant factor for the regime of noise below the threshold. Furthermore, this algorithm does not require prior knowledge of either the tree's structure, or the values of $\Delta$, $q$, $d$, or $n$.


Before presenting the result, we extend the model slightly, by allowing each node $v$ to have a distinct noise parameter $q_v$. This greater flexibility makes our results stronger. It also happens to be convenient from a technical standpoint. When $q_v$ does not depend on $v$, we say the noise is \emph{uniform}.
The following technical definition is used in our results, in place of the more crude $q \ll \frac{1}{\sqrt{\Delta}}$ given in Table \ref{table:main-walk}.
\begin{definition}\label{def:condition}
Condition \condition holds with parameter $0<\varepsilon<1$ if for every node $v$, we have
\begin{equation}\label{eq:condition}
q_v < \F{1 - \varepsilon - \Delta_v^{-\R{4}}}{\sqrt\Delta_v + \Delta_v^\R{4}}.
\end{equation}
\end{definition}
Since $\Delta_v \geq 2$, the condition  is always satisfiable when taking a small enough $\varepsilon$.

All our algorithms are deterministic, hence, expectation is taken with respect only to the sampling of the advice.
\begin{theorem}\label{thm:main-upper}
For every $\varepsilon>0$, if Condition \condition holds with parameter $\varepsilon$, then there exists a deterministic walking algorithm $\algzplus$
that requires $\bigO(\sqrt{\Delta} d)$ moves in expectation. The algorithm does not require prior knowledge of either the tree's structure, or any information regarding the values of $\Delta$, $d$, $n$, or the $q_v$'s.
\end{theorem}

In the above theorem (and some other places in this paper) the $\bigO$ notation hides terms that depends on $\varepsilon$. For Theorem~\ref{thm:main-upper}, this hidden term is $\varepsilon^{-3}$.

In Section \ref{sec:lower-exp} we establish the following lower bound.
\begin{theorem}\label{thm:main-lower}
Consider a complete $\Delta$-ary tree of depth $D$, and assume that the treasure is at a leaf. For every constant $\varepsilon>0$, if
$q\geq \F{1+\varepsilon}{\sqrt{\Delta-1}}$, then every randomized search algorithm has move (and query) complexity that in expectation is exponential in $D$. The result holds also in the full-information model. \parskip0cm
\end{theorem}
Observe that taken together, Theorems \ref{thm:main-upper}, \ref{thm:main-lower} and Condition~\condition (see Eq. \eqref{eq:condition}) imply
that for every given  $\varepsilon>0$ and large enough $\Delta$, efficient search can be achieved if $q<(1-\varepsilon)/\sqrt{\Delta}$ but not if $q>(1+\varepsilon)/\sqrt{\Delta}$.


We further complete our lower bounds with the following result, proved in Section \ref{sec:lb-query}.
\begin{theorem}\label{thm:lower-exp}
For a complete $\Delta$-ary tree of depth $D$, the expected number of queries for every algorithm is $\Omega(q\Delta D)$ (or equivalently, $\Omega(q\Delta\log_\Delta n)$).
\end{theorem}


In Section \ref{sec:PF-adv} we analyze the performance of simple memoryless algorithms called {\em probabilistic following},
suggested in \cite{eLife}. At every step, the algorithm follows the advice at the current vertex with some fixed probability $\listen$, and performs a random walk step otherwise. It turns out that such algorithms can perform well, but only in a very limited regime of noise. Specifically, we prove:
\begin{theorem}\label{thm:semi}
There exist positive constants $c_1$, $c_2$ and $c_3$ such that the following holds. If for every vertex $u$, $q_u < {c_1}/{\Delta_u}$ then there exists a probabilistic following algorithm that finds the treasure in less than $c_2 d$ expected steps.
On the other hand, if $q > {c_3}/{\Delta}$ then for every probabilistic following strategy the move complexity on a complete $\Delta$-ary tree is exponential in the depth of the tree.
\end{theorem}
Since this  algorithm is randomized, expectation is taken over both the randomness involved in sampling  advice and the possible probabilistic choices made by the algorithm.

Interestingly, when $q_u < c_1/\Delta_u$ for all vertices, this algorithm works even in the oblivious semi-adversarial model. In fact, it turns out that in the semi-adversarial model, probabilistic following algorithms are the best possible up to constant factors, as the threshold for efficient search, with respect to any algorithm, is roughly $1/ \Delta$.

\subsubsection{Results in High Probability for Walking Algorithms}\label{subsec:result-high}
We start the investigation of algorithms having a good probability guarantee with the following upper bound. The proof is presented in Section \ref{sec:walk-hp}. The $O(1)$ term in the exponent is to be understood as an absolute constant, that does not depend on either $d$ or $\eps$.
\begin{theorem}
\label{thm:walk-hp}
Let $0 < \eps < {1}/{2}$ be a constant, and suppose that $q = \Delta^{-\eps}$, and that $\Delta$ is sufficiently large ($\Delta \geq 2^{6/\eps^2}$ suffices). Let $0<\delta<1$ be a constant. Then there exists an algorithm $A'_{walk}$ in the walk model that discovers $\tau$ in $(\frac{d}{\delta})^{O(\frac 1 \eps)}$ moves with probability $1- \delta$.
Moreover, the statement holds even in the adaptive semi-adversarial variant.
\end{theorem}


\paragraph{A remark about the parameters.} {\em The restriction of $\eps < \frac{1}{2}$ is inessential to Theorem~\ref{thm:walk-hp}, and is included because the algorithms of Theorem \ref{thm:main-upper} already handle the case $\eps \ge \frac{1}{2}$.
The requirement that $\Delta$ is sufficiently large as a function of $\eps$ is natural, particularly for the semi-adversarial setting. For example, taking $\Delta \le {(3/2)}^{1/\eps}$ and keeping the assumption that $q = \Delta^{-\eps}$ will lead to  $q \ge {2}/{3}$. In the semi-adversarial setting, such levels of noise could not be overcome efficiently. To see why, consider for instance a complete binary tree. The strategy of an adversary  could be, at each faulty node, to point to a uniformly chosen neighbor, amongst the two that do not lead to $\tau$. The result would then be that  at every node, each direction of the advice is uniform,  making it useless. On the other hand, if we require $q \le \min(c, \Delta^{-\eps})$ for some suitable constant $c > 0$ that depends only on $\epsilon$, then the requirement that $\Delta$ is sufficiently large can be removed.
One can take $c = 2^{-6/\eps}$, and define $\Delta_0= 2^{6/\eps^2}$. For $\Delta \ge \Delta_0$ Theorem~\ref{thm:walk-hp} applies because $\Delta$ is sufficiently large, whereas for $\Delta \le \Delta_0$ Theorem~\ref{thm:walk-hp} applies because we may pretend that the largest degree is $\Delta_0$, and this does not affect the proofs.}\\

The upper bound shown in Theorem \ref{thm:walk-hp} is matched up to the constant in the exponent, by the following lower bound, presented in Section \ref{sec:lower-hp}.
\begin{theorem}\label{thm:lower-eps}
Let $0 < \eps < 1/2$ be an arbitrary constant, and suppose that $q = \Delta^{-\eps}$, and
that $D$ is sufficiently large, as a function of $\Delta$ and $\varepsilon$.
Consider the complete $\Delta$-ary tree of depth $D$, with the treasure placed in one of its leaves.
Let $A$ be an algorithm with success probability $1-\delta$. Then,  with constant probability, $A$ needs at least
$(\delta^{-1}D)^{\frac {1-\eps} \eps (1-o_D(*))}$ queries (and consequently also moves) before finding $\tau$.
($o_D(\cdot)$ denotes a function of $\delta$, $\epsilon$, $\Delta$ and $D$, that tends to $0$ when fixing the former parameters and letting $D$ grow to infinity.)
The statement also holds with respect to randomized algorithms.
\end{theorem}

\subsubsection{Results in the Query Model}

We start by noting that the lower bound presented in Theorem \ref{thm:main-lower}, phrased in terms of the depth $D$ of the tree, in fact already holds for query algorithms. Since the trees considered in this lower bound are complete and regular, we may replace $D$ by $\log_\Delta n$, and obtain a lower bound of  $n^{\Omega(1)}$ when $q\geq \F{1+\varepsilon}{\sqrt{\Delta-1}}$.

Concerning upper bounds on the query complexity,
we first  consider the special case that the tree is a path. For this case, we present in Section \ref{sec:path} a simple algorithm whose expected query complexity is almost tight. The path algorithm that we present and its analysis can be thought of as an adaptation of an algorithm of~\cite{bayes08} to our setting.

 Extending the path algorithm to general trees suggests the use of a separator decomposition. However, querying a separator may yield a wrong answer, and repeatedly querying a separator is useless in our model. While this can rather easily be circumvented in the case that the tree is a path, overcoming this in general trees seems to require more sophistication.
Our main idea for this purpose is to apply our walking algorithm as a subroutine to be run in the vicinity of a separator, so as to increase the reliability of detecting the direction in which to pursue the search.
Since it is easier to explain our idea in the high probability setting, we first derive a result in  that setting, and later consider results in expectation. The following is proved in Section \ref{sec:corquery}.
\begin{theorem}\label{cor:query}
Assume that $n$ is sufficiently large (as a function of $\eps$ and~$\Delta$). Under the assumptions of Theorem \ref{thm:walk-hp}, there exists an algorithm $A_{query}$ in the query model
that finds the treasure with probability  at least $1- \delta$ whose number of queries
scales like $(\delta^{-1} \log n)^{O(\varepsilon^{-1})}$.
This result holds in the adaptive semi-adversarial variant as well.
\end{theorem}

We next consider query upper bounds in expectation. The following theorem, proved in Section \ref{sec:weak-upper}, assumes that Condition \condition hold. Recall, the condition was introduced in Equation \eqref{eq:condition}, and
 can roughly be understood as saying that for all nodes $v$, $q_v < \frac{1-\eps}{\sqrt{\Delta_v}}$.

\begin{theorem}\label{thm:weak-upper}
 For every $\varepsilon>0$, there exists a  deterministic query algorithm $\algcontract$ such that if Condition \condition holds (see Eq. \eqref{eq:condition}) with parameter $\varepsilon$, then $\algcontract$ needs at most $\bigO(\sqrt{\Delta}\log \Delta \cdot \log^2 n)$ queries in expectation.
\end{theorem}
The following theorem yields a better upper bound than the one in Theorem \ref{thm:weak-upper}, but works assuming the advice parameter is the same at every node, or is bounded by the maximum degree rather than the local degree. The proof of the theorem follows the main ideas as in the proof of Theorem \ref{thm:weak-upper} but it is more technical. We therefore defer the proof to Appendix~\ref{sec:regular-upper}. Note that for constant $\Delta$ there is gap of $O(\log\log n)$ between the bound in the following theorem and that of Theorem~\ref{thm:lower-exp}. Closing this gap remains open.

\begin{theorem}\label{thm:regular-upper}
 Assume that the noise parameter at every node is bounded by $q < c/\sqrt{\Delta}$ for a sufficiently small constant $c>0$. Then there exists a deterministic query algorithm $\algregular$ for which the expected number of queries is
 $\bigO(\sqrt{\Delta}\log n \cdot \log \log n)$.
\end{theorem}



\subsection{Related Work}\label{sec:related}
In  computer science, search algorithms have been the focus of numerous works, often on a tree topology, see e.g., \cite{Laber, Newman, Onak06, Onak08}.
Within the literature on search, many works considered noisy queries \cite{Feige94, Karp07,Kempe16}.
However, it was typically assumed that noise can  be {\em resampled} at every query.
Dealing with permanent faults entails challenges that are fundamentally different from those that arise when allowing queries to be resampled. To illustrate this difference, consider the simple example  of a star graph and a constant $q<1$. Straightforward amplification can detect the target in $\bigO(1)$ queries in expectation. In contrast, in our model, it can be easily seen that $\Omega(n)$ is a lower bound for both the expected move and the query  complexities, for any constant noise parameter.

A search problem on graphs in which the set of nodes with misleading advice is chosen by an adversary was studied in \cite{Kos10,Hanusse08,Hanusse04}, as part of the more general framework of the \emph{liar models}
\cite{Aslam,BorgstromK93,CicaleseV00,Pelc02}.
Data structures with adversarial memory faults have been investigated in the so called faulty-memory RAM model, introduced in \cite{memory1}. In particular, data structures (with adversarial memory faults) that support
the same operations as search trees were studied in \cite{memory2,Brodal07}. Interestingly, the data structures developed in \cite{Brodal07} can cope with up to $O(\log n)$ faults, happening at any time during the execution of the algorithm, while maintaining optimal space and time complexity.
It is important to observe that all these models take worst case assumptions, leading to technical approaches and results which are very different from what one would expect in average-case analysis. Persistent probabilistic memory faults, as we study here, have been explicitly studied in \cite{Braverman08}, albeit in the context of sorting.
Persistent probabilistic errors were also studied in contexts of learning and optimization, see~\cite{HassidimS17}.

The noisy advice model considered in this paper actually originated in the recent biologically centered work \cite{eLife}, aiming to abstract navigation relying on guiding instructions  in the context of collaborative transport by ants.
In that work, the authors modeled ant navigation as a probabilistic following algorithm, and noticed
 that an execution of such an algorithm can be viewed as an instance of Random Walks in Random Environment (RWRE) \cite{rwreS,rwreDR}. Relying on  results from this subfield of probability theory, the authors  showed that when tuned properly, such algorithms enjoy linear move complexity on grids, provided that the bias towards the correct direction is sufficiently high.

An important theme of our work is the distinction between bounds on the expectation and bounds that hold with high probability.
When randomization is an aspect of the algorithm rather than of the input instance, there is not much of a difference between expected running time and median running time, if one is allowed to restart the algorithm several times with fresh randomness. However, there might be a substantial difference if restarting with fresh randomness is not allowed. In our model, a simple example to illustrate this phenomenon is to  consider the star graph,  and assume that a node is faulty with some small constant probability $q$. In this example, finding the treasure requires $\Omega(n)$ moves in expectation, but can be done in at most 2 moves, with probability $1-q$.
In general, for settings in which the random aspect comes up in the generation of the input instances (as in our case), generating a fresh random instance is not an option, and such a difference may or may not arise.
  In the context of designing polynomial time algorithms for distributions over instances of NP-hard problems, it is often the case that high probability algorithms are designed first, and algorithms with low expected runtime (over the same input distribution) are designed only later (see for example \cite{krivelevich-sat,Coja-Oghlan07}).

\subsection{Notation}\label{sec:notations}

Denote $p = 1-q$, and for a node $u$, $p_u = 1 - q_u$. For two nodes $u,v$, let $\pathoo{u}{v}$ denote the simple path connecting them, excluding the end nodes, and let $\pathco{u}{v}=\pathoo{u}{v}\cup \{u\}$ and $\pathcc{u}{v} = \pathco{u}{v}\cup\set{v}$.
For a node $u$, let $T(u)$ be the subtree rooted at $u$.
We denote by $\advTo{u}$ (resp. $\advAway{u}$) the set of nodes whose advice
points towards (resp. away from) $u$. By convention $u\notin\advTo{u} \cup \advAway{u}$.
Unless stated otherwise, $\log$ is the natural logarithm.

The nodes on the path from the root $\sigma$ to the treasure $\tau$ are named as $[\sigma, \tau] := \{v_0 = \sigma, v_1, \ldots, v_{d-1}, v_d = \tau\}$.
We say that node $v$ is a {\em descendant} of node $u$ if $u$ lies on the path from $\sigma$ to $v$, and $v$ is a {\em child} of $u$ if it is both a descendant of $u$ and a neighbor of $u$.


\section{Optimal Walking Algorithm in Expectation}\label{sec:walk-exp}
 In this section we prove  Theorem \ref{thm:main-upper}. At a high level, at any given time, the walking algorithm processes the advice seen so far, identifies a ``promising" node to continue from on the border of the already discovered connected component, moves to that node, and explores one of its neighbors. The crux of the matter is identifying a notion of \emph{promising} that leads to an efficient algorithm. We do so by introducing a carefully chosen \emph{prior} for the treasure location.

\subsection{Algorithm Design following a Greedy Bayesian Approach} \label{sec:intuition}
In our setting the treasure is placed by an adversary. However, we can still study algorithms that assume that it is placed according to some known distribution, and see how they measure up in our worst case setting.
A similar approach is used in \cite{bayes08}, which studies the related (but simpler) problem of search on the line. The success of this scheme highly depends on the choice of the prior distribution we take.

Suppose first that the structure of the tree is known to the algorithm, and that the treasure is placed according to some known distribution $\theta$ supported on the leaves. Let $\adv$ denote the advice on the nodes we have already visited. Aiming to find the treasure as fast as possible, a possible greedy algorithm explores the vertex that, given the advice seen so far, has the highest probability of having the treasure in its subtree.

We extend the definition of $\theta$ to internal nodes by defining $\theta(u)$ to be the sum of $\theta(w)$ over all leaves $w$ of $T(u)$. Given some $u$ that was not visited yet, and given the previously seen advice $\adv$,
the probability of the treasure being in $u$'s subtree  $T(u)$,  is:
\begin{align*}
\cprob{\tau \in T(u)}{\adv}
&= \F{\prob{\tau \in T(u)}}{\prob{\adv}} \cprob{\adv}{\tau \in T(u)}
\\ & = \F{\theta(u)}{\prob{\adv}}
 \prod_{w \in \advTo{u}} \B{p_w + \F{q_w}{\Delta_w}}
 \prod_{w \in \advAway{u}} \F{q_w}{\Delta_w}.
\end{align*}
The last factor is $q_w/\Delta_w$ because it is the probability that the advice at $w$ points exactly the way it does in $\adv$, and not only away from $\tau$.
Note that the advice seen so far
does not involve vertices in $T(u)$, because we consider a walking algorithm, and $u$ has not been visited yet.
Therefore, if $\treasure\in T(u)$ then correct advice in $\adv$ points to $u$. We ignore the term $p_w + q_w/\Delta_w$ because this will not affect the results by much, and applying a log
we can approximate the relative strength of a node by:
$$
\log(\theta(u)) + \sum_{w \in \advAway{u}} \log\BF{q_w}{\Delta_w}.
$$
We replace $q_w$ by its upper bound $1/\sqrt\Delta_w$ (consequently, the algorithm need not know the exact value of $q_w$). After scaling, we obtain our definition for the score of a vertex:
$$
\score(u) = \F{2}{3}\log(\theta(u)) - \sum_{w \in  \advAway{u}} \log(\Delta_w).
$$
When comparing two specific vertices $u$ and $v$, $\score(u) > \score(v)$ iff:
$$
\sum_{\substack{w \in \pathoo{u}{v}  \cap  \advTo{u}}} \log(\Delta_w)
-
\sum_{\substack{w \in \pathoo{u}{v} \cap  \advTo{v}}} \log(\Delta_w)
> \F{2}{3}\log\BF{\theta(v)}{\theta(u)}.
$$
This is because any advice that is not on the path between $u$ and $v$ contributes the same to both sides, as well as advice on vertices on the path that point sideways, and not towards $u$ or $v$.
Since we use this score to compare two vertices that are neighbors of already explored vertices, and our algorithm is a walking algorithm, then we will always have all the advice on this path.
In particular, the answer to whether $\score(u) > \score(v)$, does not depend on the specific choices of the algorithm, and does not change throughout the execution of the algorithm, even though the scores themselves do change.
The comparison depends only on the advice given by the environment.

Let us try and justify the score criterion at an intuitive level. Consider the case of a complete $\Delta$-ary tree, with $\theta$ being the uniform distribution on the leaves.
Here $score(u) > score(v)$ if (cheating a little by thinking of $\log(\Delta)$ and $\log(\Delta - 1)$ as equal):
$$
\bigl|\advTo{u} \cap \pathoo{u}{v}\bigr| -
\bigl|\advTo{v}  \cap \pathoo{u}{v}\bigr|
>
\F{2}{3}\bigl(d(u) - d(v)\bigr).
$$
If, for example, we consider two vertices $u,v \in T$ at the same depth, then $score(u) > score(v)$ if there is more advice pointing towards $u$ than towards $v$. If the vertices have different depths, then the one closer to the root has some advantage, but it can still be beaten.

For general trees, perhaps the most natural $\theta$ is the uniform distribution on all nodes (or just on all leaves - this choice is  actually similar). It is also a generalization of the example above.
Unfortunately, however, while this works well on the complete $\Delta$-ary tree,  the uniform distribution fails on other (non-complete) $\Delta$-ary trees (see a preliminary version of this work \cite{advice1} for details).

\subsection{Algorithm $\algzplus$}
In our context, there is no distribution over treasure location and we are free to choose  $\theta$ as we like. Take $\theta$ to be the distribution defined by a simple random process. Starting at the root, at each step, walk down to a child uniformly at random, until reaching a leaf. For a leaf $v$, define $\theta(v)$ as the probability that this process eventually reaches $v$. Our extension of $\theta$ can be interpreted as $\theta(v)$ being the probability that this process passes through $v$. Formally, $\theta(\sigma)=1$, and
$\theta(u) = (\Delta_\sigma \prod_{w \in \pathoo{\sigma}{u}} (\Delta_w - 1))^{-1}$.
It turns out that this choice, slightly changed, works remarkably well, and
gives an optimal algorithm in noise conditions that  practically match those of our lower bound.
For a vertex $u \neq \sigma$, define:
\begin{equation}\label{eq:beta}
\beta(u) = \prod_{w \in \pathco{\sigma}{u}} \Delta_w.
\end{equation}
It is a sort of approximation of $1/\theta(u)$, which we prefer for technical convenience. Indeed,
 for all $u$, $1/\beta(u) \leq \theta(u)$.
 A useful property of this $\beta$ (besides the fact that it gives rise to an optimal algorithm) is that to calculate $\beta(v)$ (just like $\theta$), one only needs to know the degrees of the vertices from $v$ up to the root.

In the walking algorithm, if $v$ is a candidate for exploration, the nodes on the path from $\sigma$ to $v$ must have been visited already and so the algorithm does not need any a priori knowledge of the structure of the tree.
 The following claim will be soon useful:
\begin{claim}\label{claim:exp-leaves}
The following two inequalities hold for every $c<1$:
$$
\sum_{v \in T} \F{c^{d(v)}}{\beta(v)} \leq \R{1-c},~~\, \,
\sum_{v \in T} \F{d(v)c^{d(v)}}{\beta(v)} \leq \F{c}{(1-c)^2}.
$$
\end{claim}
\begin{proof}
To prove the first inequality,
follow the same random walk defining $\theta$. Starting at the root, at each step choose uniformly at random one of the children of the current vertex. Now, while passing through a vertex $v$ collect $c^{d(v)}$ points.
No matter what choices are made, the number of points is at most $1 + c + c^2 + ...= 1/(1-c)$. On the other hand, $\sum_{v \in T} \theta(v)c^{d(v)}$ is
the expected number of points gained. The result follows since $1/\beta(v) \leq \theta(v)$.
The second inequality is derived similarly, using the fact that
$c + 2c^2 + 3c^3 + \ldots = c/(1-c)^2$.
\end{proof}
For a vertex $u\in T$ and previously seen advice $\adv$ define:
\begin{equation}\label{eq:score}
\score(u) = \F{2}{3}\log\BF1{\beta(u)}
-   \sum_{\substack{w \in \advAway{u}  }} \log(\Delta_w).
\end{equation}
This is similar to the definition of $\score(u)$ given in Section~\ref{sec:intuition}, except that $\theta(u)$ is now replaced by its approximation $\frac{1}{\beta(u)}$.

Algorithm $\algzplus$ keeps track of all vertices that are children of the vertices it explored so far, and repeatedly walks to and then explores the one with highest score according to the current advice, breaking ties arbitrarily.
As stated in the introduction, the algorithm does not require prior knowledge of either the tree's structure, or the values of $\Delta$, $q$, $D$ or $n$.
\subsection{Analysis}

Recall the definition of Condition \condition from Equation \eqref{eq:condition}. The next lemma provides a large deviation bound tailored to our setting.

\begin{lemma}\label{lem:maintec}
Consider independent random variables $X_1, \ldots, X_\ell$, where $X_i$ takes the values $(-\log \Delta_i, 0, \log \Delta_i)$ with respective probabilities $( \frac{q_i}{\Delta_i}, q_i(1-\frac{2}{\Delta_i}),p_i+\frac{q_i}{\Delta_i})$, for parameters $p_i,q_i = 1-p_i$ and $\Delta_i >0$. Assume that  Condition~\condition holds for some $\varepsilon>0$. Then for every integer (positive or negative) $m$,
$$
\prob{\sum_{i = 1}^{\ell} X_i \leq m}
\leq
 e^\F{3m}{4} \constq^\ell \prod_{i = 1}^{\ell} \R{\sqrt{\Delta_i}}.
$$
\end{lemma}
\begin{proof}
For every $s \in \RR$,
\eq{
\prob{\sum_{i=1}^{\ell} X_i \leq \thresh}
& =
\prob{e^{s\sum_{i=1}^{\ell} -X_i} \geq e^{-s\thresh}}
\leq
{e^{s\thresh}}{\expct{e^{s\sum_i -X_i}}}
=
{e^{s\thresh}}{\prod_i \expct{e^{-sX_i}}}
\\ & =
e^{s\thresh} \prod_{i=1}^{\ell} \B{\F{p_i+\F {q_i} {\Delta_i}}{e^{\log(\Delta_i)s}} + q_i\B{1-\F 2{\Delta_i}} + \F{q_i}{\Delta_i}e^{\log(\Delta_i)s}}
\\ &
\leq
{e^{s\thresh}} \prod_{i=1}^{\ell} \B{\R{\Delta_i^s} + q_i + q_i \Delta_i^{s-1}}.
}
We take $s = \F34$, and get:
\[\prob{\sum_{i=1}^{\ell} X_i \leq \thresh} \leq {e^\F{3\thresh}{4}} \prod_{i=1}^{\ell} \B{\Delta_i^{-\F34} + q_i + q_i\Delta_i^{-\R4}} \leq {e^\F{3\thresh}{4}} \prod_{i=1}^{\ell} \F{1 - \varepsilon}{\sqrt{\Delta_i}},
\]
where for the last step we used Condition ($\star$) which says $q_i < \F{1 - \varepsilon - \Delta_i^{-\R{4}}}{\sqrt\Delta_i + \Delta_i^\R{4}}$ implying that:
\eq{
q_i \Delta_i^\R2 + q_i \Delta_i^\R4  + \Delta_i^{-\R4} < 1 - \varepsilon \mbox{,~~and hence}\\
\Delta_i^{-\F34} + q_i + q_i\Delta_i^{-\R4} < \F{1 - \varepsilon}{\sqrt{\Delta_i}}.
}
\end{proof}
The next theorem 
establishes  Theorem \ref{thm:main-upper}.
\begin{theorem}\label{thm:zplus}
Assume that Condition~\condition holds for some fixed $\varepsilon>0$. Then $\algzplus$ requires only $\bigO(d\sqrt{\Delta})$ moves in expectation. The constant hidden in the $\bigO$ notation only depends polynomially on $1/\varepsilon$.
\end{theorem}
\begin{proof}
Denote the vertices on the path from $\source$ to $\treasure$ by $\sigma = u_0, u_1, \ldots, u_d=\treasure$ in order. Denote by $E_k$ the expected number of moves required to reach $u_{k}$ once $u_{k-1}$ is reached. We will show that for all $k$,  $E_k = \bigO(\sqrt \Delta)$, and by linearity of expectation this concludes the proof.

Once $u_{k-1}$ is visited, $\algzplus$ only goes to some of the nodes that have score at least as high as  $u_k$. We can therefore bound $E_k$ from above by assuming we go through all of them, and this expression does not depend on the previous choices of the algorithm and the nodes it saw before seeing $u_k$.
The length of this tour is bounded by twice the sum of distances of these nodes from $u_k$. Hence,
$$
E_k  \leq 2\sum_{i=1}^k
\sum_{u \in C(u_i)}
\prob{\score(u) \geq \score(u_k)} \cdot \dist(u_k, u),
$$
where $C(u_k) = T(u_{k-1}) \setminus T(u_{k})$, and so $\bigcup_{i=1}^k C(u_i) = T \setminus T(u_k)$.
Recall from Eq.~\eqref{eq:score} that scores are defined so that
\begin{align*}
\score(u_k)  \leq \score(u)
\\ \Longleftrightarrow \\
\sum_{w \in  \advAway{u}} \log(\Delta_w) -
\sum_{w \in  \advAway{u_k}} \log(\Delta_w)
\leq
\F{2}{3}\B{\log\BF1{\beta(u)} -  \log\BF1{\beta(u_k)}}
\\ \Longleftrightarrow \\
\sum_{w \in \pathoo{u}{u_k}} \begin{cases}
\log(\Delta_w) & w \text{ points towards } u_k\\
-\log(\Delta_w) & w \text{ points towards } u \\
0 & \text{otherwise} \\
\end{cases}
\ \
\leq
\F{2}{3}\log\BF{\beta(u_k)}{\beta(u)}
\end{align*}
Indeed, a vertex $x$ on the path should point towards $u_k$: this happens with probability $p_x + q_x/\Delta_x$. Otherwise, it points towards $u$ with probability $q_x/\Delta_x$, and elsewhere with probability $q_x(1 - 2/\Delta_x)$.
Denoting $c = 1 - \varepsilon$,
setting $m=\F{2}{3}\log(\beta(u_k)/\beta(u))$, and applying Lemma \ref{lem:maintec} we can upper bound the probability of this happening:
\eq{
\F {E_k} 2 &\leq
\sum_{i=1}^k
\sum_{u \in C(u_i)}
e^{\F{3}{4}\cdot\F{2}{3}\log\BF{\beta(u_k)}{\beta(u)}}\cdot {c^{\dist(u_k, u)-1}}
\sqrt{
\prod_{v \in \pathoo{u}{u_k}} \R{\Delta_v}
}\cdot \dist(u_k,u)
\\ & =
\F{1}{c}\sum_{i=1}^k
\sum_{u \in C(u_i)}
\F{c^{\dist(u_k, u)}}{\sqrt{\F{\beta(u)}{\beta(u_k)}}}
\sqrt{\F{\Delta_{u_i}}{
\F{\beta(u_k)}{\beta(u_i)} \cdot \F{\beta(u)}{\beta(u_i)}
}} \cdot \dist(u_k,u)
\\ & \leq
\F{\sqrt{\Delta}}{c}
\sum_{i=1}^{k}
c^{\dist(u_k,u_i)}
\sum_{u \in C(u_i)}
c^{\dist(u_i,u)}\F{\beta(u_i)}{\beta(u)} \cdot \Big(\dist(u_k,u_i) + \dist(u_i,u)\Big)
.}
By Claim \ref{claim:exp-leaves}, applied to the tree rooted at $u_i$, we get:
\[\sum_{u \in C(u_i)} \F{c^{\dist(u_i,u)} \beta(u_i)}{\beta(u)} < \R{1-c}, ~~~~\mbox{and}~~~~
\sum_{u \in C(u_i)} \F{c^{\dist(u_i,u)} \beta(u_i)}{\beta(u)} \dist(u_i,u) < \F{c}{(1-c)^2}.\]
And so:
\eq{
\F {E_k} 2 &\leq
\F{\sqrt\Delta}{c(1-c)}
\sum_{i=1}^{k}
c^{\dist(u_k,u_i)}\dist(u_k,u_i)
+\F{\sqrt\Delta}{(1-c)^2}
\sum_{i=1}^{k}
c^{\dist(u_k,u_i)}
\\ & \leq
\F{(1+c)\sqrt\Delta}{(1-c)^3}
\leq \F{2\sqrt{\Delta}}{\varepsilon^3}
= \bigO\left( \sqrt{\Delta}\right),}
where we again used the equality $c + 2c^2 + 3c^3 + \ldots = c/(1-c)^2$.
\end{proof}

\section{Lower bounds in Expectation}\label{sec:lower-exp}

Several of our lower bounds will invoke the following straightforward observation (whose proof we omit).
\begin{observation}\label{choosing}
Any randomized algorithm trying to find a treasure chosen uniformly at random between $k$ identical objects will need an expected number of queries that is at least $(k+1)/2$.
\end{observation}

\subsection{The random noise model}

We prove here Theorem \ref{thm:main-lower}. Namely, that
for every fixed $\varepsilon>0$, and for every complete $\Delta$-ary tree, if $q\geq \F{1+\varepsilon}{\sqrt{\Delta-1}}$, then every randomized search algorithm has query (and move) complexity which is exponential in the depth $D$ of the treasure.

\begin{proof}
Our lower bound holds also in the query model, as we assume that the algorithm gets as input the full topology of the tree. Moreover, to simplify the proof, we give the algorithm access to additional information, and prove the lower bound even against this stronger algorithm.
The algorithm is strengthened in two respects: for every internal (non-leaf) node, the algorithm is told whether the node is faulty, and moreover, if the internal node is non-faulty, the advice of the node is revealed to the algorithm. This information for all internal nodes is revealed to the algorithm for free, without the algorithm making any query. Given that the algorithm is told which nodes are faulty, we may assume that faulty nodes have no advice at all, because faulty advice is random and hence can be generated by the algorithm itself.

Given the pattern of faults, a leaf is called {\em uninformative} if the whole root to leaf path is faulty. Denote the number of leaves by $N = \Delta (\Delta-1)^{D-1}$, and the expected number of uninformative leaves by $\mu = N q^D$. Let $p_i$ denote the probability that there are exactly $i$ uninformative leaves.

The adversary places the treasure at a random leaf. Conditioned on there being $i$ uninformative leaves, the probability of the treasure being at an uninformative leaf is exactly $\frac{i}{N}$. If the treasure is located at an uninformative leaf, the advice of all nonfaulty internal nodes points to the root (recall that there is no advice on faulty nodes), and the algorithm can infer that the treasure is at an uninformative leaf. As all uninformative leaves are equally probable, the expected number of queries that the algorithm needs to make is exactly $\frac{i+1}{2}$. Hence the expected number of queries (in this stronger model) is:

$\sum_i \frac{i+1}{2} \frac{i}{N} p_i > \frac{1}{2N} \sum_i i^2 p_i \ge \frac{1}{2N}\mu^2 = \frac{Nq^{2D}}{2}$

\noindent In the last inequality we used the fact that $E[X^2] \ge (E[X])^2$ for every random variable $X$. (In our use the random variable is $i$, the number of uninformative leaves.)

Hence if $q\geq \F{1+\varepsilon}{\sqrt{\Delta-1}}$, the expected number of queries is at least $\frac{1}{2}(1+\varepsilon)^{2D}$.
\end{proof}

\subsection{The Semi-Adversarial Variant}\label{sec:adv_lb}

Recall that in contrast to the purely probabilistic model, in the oblivious semi-adversarial model,
a faulty node $u$ no longer points to a neighbor chosen uniformly at random, and instead, the neighbor $w$ which such a node points to is chosen by an adversary. Importantly, for each node $u$, the adversary must specify its potentially faulty advice $w$, before it is known which nodes will be faulty. In other words, first, the adversary specifies the faulty advice $w$ for each node $u$, and then the environment samples which node is faulty.

In the semi-adversarial noise model, if $q>1/(\Delta-1)$ then any algorithm must have expected query and move complexity that are exponential in the depth $D$.

\begin{theorem}\label{thm:lb_adv}
Consider an algorithm in the oblivious semi-adversarial model.
On the complete $\Delta$-ary tree of depth $D$, the expected number of queries to find the treasure is $\Omega\left( (q \Delta)^D \right)$.
The lower bound holds even if the algorithm has access to the advice of all internal nodes in the tree.
\end{theorem}

\begin{proof}
Consider the complete $\Delta$-ary tree and restrict attention to the cases where the treasure is located at a leaf. The adversary behaves as follows. For every advice it gets a chance to manipulate, it always make it point towards the root.
With probability $q^D$ the adversary gets to choose all the advice on the path between the root and the treasure. Any other advice points towards the root as well (either because it was correct to begin with or because it was set by the adversary).
Hence with probability $q^D$ the tree that the algorithm sees is the same regardless of the position of the treasure. When this happens, the expected time to find the treasure is $\Omega(\Delta^D)$ (linear in the number of leaves), by Observation~\ref{choosing}.
\end{proof}

\section{Probabilistic Following Algorithms}\label{sec:PF-adv}

In this section we present our results on the probabilistic following algorithms described in the introduction. As mentioned, such algorithms can perform well also in the more difficult oblivious semi-adversarial setting. 

 Recall that a {\em Probabilistic Following} ($\RL$) algorithm is specified by a {\em listening} parameter $\listen \in (0,1)$. At each step, the algorithm ``listens'' to the advice with probability $\listen$ and takes a uniform random step otherwise. The first item in the next theorem states that if the noise parameter is smaller than $c/\Delta$ for some small enough constant $0<c<1$, then there exists a listening parameter $\listen$ for which Algorithm $\RL$  achieves $\bigO(D)$ move complexity. Moreover, this result holds also in the oblivious semi-adversarial model. Hence, together with Theorem \ref{thm:lb_adv}, it implies that in order to achieve efficient search, the noise parameter threshold for the semi-adversarial model is $\Theta(1/\Delta)$.

\begin{theorem}\label{thm:rw}
\begin{enumerate}
    \item
    Assuming $q_u < 1/(10\Delta_u)$ for every $u$, then $\RL$ with parameter $\lambda = 0.7$ finds the treasure in less than $100 D$ expected steps, even in the oblivious semi-adversarial setting.
    \item
    Consider the complete $\Delta$-ary tree and assume that $q> {10}/{\Delta}$. Then for every  choice of $\lambda$ the hitting time of the treasure by $\RL$ is exponential in the depth of the tree, even assuming the faulty advice is drawn at random.
\end{enumerate}
\end{theorem}
\begin{proof}
Our plan is to show that the expected time to make one step in the correct direction is $\bigO(1)$, from any starting node. Conditioning on the advice setting, we make use of the Markov property to relate these elementary steps to the total travel time. The main delicate point in the proof stems from dealing with two different sources of randomness. Namely the randomness of the advice and that of the walk itself.

In this section, it is convenient to picture the tree as rooted at the target node $\treasure$. For every node $u$ in the tree, we denote by $u'$ the parent of $u$ with respect to the treasure. With this convention, correct advice at a node $u$ points to $u'$, while incorrect advice points to one of its children.
The fact the walk moves on a tree means that for a given advice setting, the expected (over the walk) time it takes to reach $u'$ from $u$ can be written conveniently as a product of two independent random variables: one random variable that depends only on the advice at $u$, and the other depends only on the advice on the set of $u$'s descendants.

We denote by $\reach(u)$ the time it takes to reach node $u$. We also introduce a convention regarding the notation used to denote expectations.
In our setting there are two sources of randomness, the first being the randomness used in drawing the advice, and the second being the randomness used in the walk itself.
We write $\EE$ for expectation taken over the advice, while we use $\QE[u]$ to denote expectation over the walk, conditioning on $u$ being the starting node.
Observe that $\QE[u](\reach(v))$ is a random variable with respect to the advice, whereas $\EE \QE[u](\reach(v))$ is just a number.

The following is the central lemma of this section.
\begin{lemma}\label{lem:ind}
Assume that for every vertex $u$, $q_u  < 1/(10\Delta_u)$, and $\lambda =0.7$. Then for all nodes $u$, $\EE \QE[u]\reach(u') \leq 100$. The result holds also in the oblivious semi-adversarial model.
\end{lemma}

Let us now see how we can conclude the proof of the first item in Theorem \ref{thm:rw}, given the lemma. Consider a designated source $\source$.
Let us denote by $\source=u_d,u_{d-1}, \ldots, u_0 = \target$ the nodes on the path from $\source$ to $\treasure$.
Let $\period_i$ be
the random variable indicating the time it takes to reach $u_{i-1}$ after $u_i$
has been visited for the first time.
The  time to reach $\treasure$ from $\source$ is precisely $\sum_{i=1}^{\dist(\source,\treasure)} \period_i$.
Hence, the expected  time to reach $\treasure$ from $\source$ is $\sum_{i=1}^{\dist(\source,\treasure)} \EE[\QE[\source]\period_i]~.$
Conditioning on the advice setting, the process is a Markov chain and we may write
\eq{
\QE[\source]\period_i = \QE[u_{i}]\reach(u_{i-1}).}
Taking expectations over the advice ($\EE$),
under the assumptions of Lemma \ref{lem:ind}, it follows that
$ \EE(\QE[\source]\period_i) \leq 100$, for every $i\in[\dist(\source,\treasure)]$. And this immediately implies a bound of $100 \cdot \dist(\source,\treasure)$.

\begin{proof}[Proof of Lemma \ref{lem:ind}]
We start with partitioning the nodes of the tree according to their distance from the root~$\treasure$.
More precisely, for $i=1,2,\ldots, \depth$, where $\depth$ is the depth of the tree, let us define
 $$
 \Lvl_i := \{ u \in T: \dist(u,\treasure)=i \}~.
 $$
The nodes in $\Lvl_i$  are referred to as {\em level-$i$} nodes.
We treat the statement of the lemma for nodes $u \in \Lvl_i$ as an induction hypothesis, with $i$ being the induction parameter.
The induction goes backwards, meaning we assume the assumption holds at level $i+1$ and show it holds at level $i$. The case of the maximal level (base case for the induction) is easy since, at a leaf the walk can only go up and so if $u$ is a leaf $\EE \QE[u](\reach(u')) = 1 < 100$.

 Assume now that $u \in \Lvl_i$. We first condition on the advice setting.
A priori, $\QE[u]t(u')$ depends on the advice over the full tree, but in fact it is easy to see that only advice at layers $\geq i$ matter. Recall from Markov Chain theory that an \emph{excursion} to/from a point is simply the part of the walk  between two visits to the given point.
We denote $L_u$ the average (over the walk only) length of an excursion from $u$ to itself that does not go straight to $u'$, and we write $N_{u}$ to denote the expected (over the walk only) number of excursions before going to $u'$. We also refer to this number as a number of \emph{attempts}. The variable $N_{u}$ can be $0$ if the walk goes directly to $u'$ without any excursion.
We decompose $\reach(u')$ in the following standard way, using the Markov property
\eql{
\QE[u]\reach(u') = 1+ L_u \cdot N_{u}.
}{eq:firstWald}
Indeed, the expectation $\QE[u]\reach(u')$ can be seen as the expectation (over the walk) of $1+\sum_{i=1}^T Y_i$ where the $Y_i$'s are the lengths of each excursion from $u$ and $T$ is the (random) number of such excursions before hitting $u'$. The term $1+$ accounts for the step from $u$ to $u'$. The event $\{T \geq t \}$ is independent of $Y_1, \ldots, Y_t$ and so using
Wald's identity we have that
$\QE[u]\reach(u')  = 1+\QE[u] T \cdot \QE[u] Y_1$. The term $\QE[u] T$ is equal to $N_u$ (by definition) while $\QE[u] Y_1$ is equal to $L_u$ (by definition).

We now want to average equality \eqref{eq:firstWald}, which is only an average over the walk, by taking the expectation over all advice in layers $\geq i$. To this aim, we write $L_u$ as follows
\eq{
L_u = 1+\sum_{v\neq u', v\sim u} p_{u,v}\QE[v] \reach(u),
}
where we write $u\sim v$ when $u$ and $v$ are neighbors in the tree and $p_{u,v}$ is the probability to go straight from $u$ to $v$ given the advice setting.
By assumption on the model, $\QE[v]\reach(u)$ depends on the  advice at layers $\geq i+1$ only, if we start at a node $v \in \Lvl_{i+1}$, while
both $p_{u,v}$ and $N_{u}$ depend only on the advice at layer $=i$ of the tree. This is true also in the oblivious semi-adversarial model.
Hence when we average, we can first average over the advice in layers $> i$ to obtain, denoting $\EE^{>i}$, the expectation over the layers $>i$,
\eql{
\EE^{>i}\QE[u]\reach(u') &= 1+ \left( 1+\sum_{v\neq u', v\sim u} p_{u,v} \EE^{>i}\QE[v]\reach(u)\right) N_{u},\nonumber\\
&= 1+ \left( 1+\sum_{v\neq u', v\sim u} p_{u,v} \EE\QE[v]\reach(u)\right) N_{u}.
}{eq:decomptau}
and using the fact that,$
\sum_{v \neq u'} p_{u,v} \leq 1,$
together with the induction assumption at rank $i+1$,
we obtain
\eq{
\EE^{>i}\QE[u]\reach(u') & \leq
1+ \left( 1+100\right) N_{u}.
}
Averaging over the layer $i$ of advice we obtain
\eq{
\EE\QE[u]\reach(u') & \leq
1+ 101 \EE N_{u}.
}
It only remains to analyse the term $\EE N_{u}$. If the advice at $u$ is correct, which happens with probability $p_u = 1 - q_u$, then the number of attempts follows a (shifted  by $1$) geometric law with parameter
$\listen +\frac{ (1-\listen)}{\Delta_u}$. In words, when the advice
points to $u'$ which happens with probability at most $1$, the walker can go to the correct node either because she listens to the advice, which happens with probability $\listen$, or because she did not listen, but still took the right edge, which happens with probability $\frac{ (1-\listen)}{\Delta_u}$. Similarly, when the advice points to a node $\neq u'$, which happens with probability at most $q_u$, then $N_{u}$ follows a geometric law (shifted by $1$) with parameter $\frac{ (1-\listen)}{\Delta_u}$. The conclusion is that
\eql{
\EE N_{u} &\leq  \left(\frac{1}{\listen +\frac{ (1-\listen)}{\Delta_u}}-1\right)
+ q_u \left(\frac{\Delta_u}{1-\listen}-1\right)
\nonumber \\ &\leq
\frac{1}{\listen} - 1 + \frac{q_u \Delta_u}{1-\listen}
}{eq:nu}
And so it follows that
\eq{
\EE\QE[u]\reach(u') \leq
1+ 101 \cdot
\left(\frac{1}{\listen} - 1 + \frac{q_u \Delta_u}{1-\listen} \right)
}
Hence if $q_u\Delta_u < 0.1$ and we choose $\lambda = 0.7$ (this choice is a tradeoff between two considerations: $\lambda$ bounded away above $\frac{1}{2}$ is required so that in expectation we make constant progress towards the treasure when the advice is correct, and $\lambda$ bounded away below~1 is required so as to have probability larger than $q$ for advancing when the advice is incorrect), we see that $\EE N_{u}<0.9 $. This is because
\[
\frac{1}{\listen} - 1 + \frac{0.1}{1-\listen}
\leq
\F{10}{7} - 1 + \frac{0.1}{1-0.7}  < 0.9
\]
Hence it follows that $
\EE\QE[u]\reach(u') \leq
1+0.9 \cdot 101< 100.$
By our (backwards) induction, we have just shown that, if $q< \frac{1}{10\Delta}$ and we set $\lambda = 0.7$ then
for all nodes $u$ in the tree
\eq{
\EE\QE[u]\reach(u') < 100.
}
This concludes the proof of Lemma \ref{lem:ind} and hence also of the first part of Theorem \ref{thm:rw}.
\end{proof}

Let us explain how the lower bound in the second part of Theorem \ref{thm:rw} is derived in the case that $q\Delta>10$.
We assume the complete $\Delta$-ary tree under our the uniform noise model.
With probability $q$ there is fault at $u$ and with probability $1-\F 1 \Delta$ the advice does not point to $u'$.
Recall that $N_{u}$ denotes the expected (over the walk only) number of excursions starting from $u$ before going to $u$'s parent.
Then, $N_u$ follows a geometric law with parameter $\F {1-\lambda}{\Delta}$.
Hence
\begin{align*}
\EE(N_u) \geq q\Delta\left(1-\F 1 \Delta\right)  \frac{1}{1-\listen} - 1 \geq \F{10 (1-\F 1 \Delta)}{1-\listen} -1\geq 10\left(1-\F 1 \Delta\right) -1\geq 3,
\end{align*}
for every choice of $\lambda$, since $\Delta\geq 2$.
We proceed by induction similar to the proof of the first part of Theorem \ref{thm:rw}, and use Equality (\ref{eq:decomptau}) together with the lower bound on $\EE(N_u)$ to obtain that for every node on layer $i$, $u$ with parent $u'$,
$\EE \QE[u]\reach(u') \geq 1+3 \min_{v \in \Lvl_{i+1}}\EE \QE[v]\reach(v')$, so in particular
\begin{align*}
\min_{u \in \Lvl_i}\EE \QE[u]\reach(u') \geq 1+3 \min_{v \in \Lvl_{i+1}}\EE \QE[v]\reach(v').
\end{align*}
The expected hitting time of the target $\target$, even starting at one of its children is therefore of order $\Omega(3^D)$.
\end{proof}

%

\section{Walking Upper Bounds in High Probability}\label{sec:walk-hp}

Our goal in this section is to prove Theorem \ref{thm:walk-hp}, stating that for any constants $0<\delta<1$ and $0 < \eps < {1}/{2}$, if $q = \Delta^{-\eps}$, and  $\Delta$ is sufficiently large (specifically, $\Delta \geq 2^{6/\eps^2}$, see remark in Section \ref{subsec:result-high}), then there exists a walking algorithm $A'_{walk}$ (parameterized by $\delta$, $\Delta$ and $\eps$) that  discovers $\tau$ in $({d}/{\delta})^{O(\frac 1 \eps)}$ moves with probability $1- \delta$.
Moreover, the statement holds even in the adaptive semi-adversarial variant. 

\subsection{The Meta Algorithm}\label{sec:meta}
Underlying the upper bound presented in Theorem \ref{thm:walk-hp} is a simple, yet general, scheme. It is based on a binary notion of \emph{fitness}, which we define later. This notion depends on the parameters of the model. It is carefully crafted such that the following {\em fitness properties} hold:
\begin{itemize}
\item {F1.} Whether or not a node $u$ is fit only depends on the advice on the path $[\sigma, u]$, excluding $u$.\label{bullet:fitness}
\item {F2.} For every node $u$ on the path $[\sigma, \tau]$,
$\Pr(u \mbox{ is fit})\geq 1-\frac{\delta}{2D}.
$
\label{bullet:correct}
\item {F3.} With probability at least $\geq 1-\frac{\delta}{2}$, the connected component of fit nodes that contains the root is of size bounded by $f(D, \delta)$, for some function $f$.
\label{bullet:fast}
\end{itemize}

Once fitness is appropriately defined so that properties  F1 - F3 hold, a {\em depth first search} algorithm can be applied in the walk model. It consists of exploring in a depth-first fashion the connected component of fit nodes containing the root. We refer to this algorithm as $A'_{walk}$.
Property F1 ensures that $A'_{walk}$ is well-defined.
The time to explore a component is at most twice its size, because each edge is traversed at most twice.

\begin{claim}\label{claim:f2}
Property F2 implies that $A'_{walk}$ eventually finds $\tau$ with probability $\geq 1- \frac{\delta}{2}$.
\end{claim}
\begin{proof}
Using Property F2, the probability that all nodes on the root to treasure path $[\sigma, \tau]$ are fit is  at least as large as
 $1-\frac{\delta D}{2D} = 1-\frac{\delta}{2}$.  Under this event, $\tau$ belongs to the same component of fit nodes as the root, and hence $A'_{walk}$ eventually finds it.
\end{proof}
By Property F3, the $A'_{walk}$ algorithm needs a number of steps which is upper bounded by $2f(D, \delta)$ with probability $1- \frac{\delta}{2}$.
Using a union bound we derive the following claim.

\begin{claim}\label{claim:meta}
If the fitness is defined so that properties F1-F3 are satisfied, then $A'_{walk}$ finds $\tau$
in at most $2f(D, \delta)$ steps with probability $\geq 1- \delta$.
\end{claim}

\subsection{Upper Bound in the Walk Model with High Probability}\label{sec:up-walk}
This subsection is devoted to the proof of Theorem \ref{thm:walk-hp}. We assume the following, w.l.o.g.

\begin{itemize}

\item The noise model is the adaptive semi-adversarial variant. Hence the results apply also to the oblivious semi-adversarial and to the random variants.

\item The algorithm knows the depth $d$ of the treasure.
This assumption can be removed by an iterative process that guesses the depth to be 1,2,\ldots. By running each iteration for a limited time as specified in the theorem,  the asymptotic runtime is not violated.

    Given that the algorithm knows the depth $d$ of the treasure, we further assume w.l.o.g.~that it never searches a node at depth greater than $d$. Equivalently, we may (and do) assume that the depth of the tree is $d=D$, i.e., that the treasure is located at a leaf.

\item The tree is balanced: all leaves are at depth $D$, and all non-leaf nodes have degree exactly $\Delta$. To remove this assumption,
whenever the algorithm visits a node $v$ at depth $i < D$ of degree $d_v < \Delta$, it can connect to it $\Delta - d_v$ ``auxiliary trees", where each auxiliary tree has depth $D - i - 1$ and is balanced. The advice in all nodes of these auxiliary trees points towards $v$, which is a valid choice in the adaptive semi-adversarial model.

\end{itemize}
Our algorithm $A'_{walk}$ follows the general scheme presented in Section \ref{sec:meta}.
It is based on a notion of fitness presented below.
With this notion in hand, $A'_{walk}$ simply visits, in a depth-first fashion, the component of fit nodes to which the root $\sigma$ belongs.
\begin{definition}\label{def:fit}[Advice-fitness]
Let $\ldepth_1 = \frac{2}{\eps} \log_{\Delta} (4\delta^{-1} D )$
and $\ldepth_2 = \frac{6}{\eps^2}\log_{\Delta} (4\delta^{-1} D)$.
Let $u$ be a node and $u_{-\ldepth_2}$ be the ancestor of $u$ at distance $\ldepth_2$ from $u$, or $\sigma$ if $u$ is at distance $< \ldepth_2$ from $\sigma$.
The node $u$ is said to be \emph{fit} if the number of locations  on the path $[u_{-\ldepth_2}, u]$ that do not point towards $u$ is less than $\ldepth_1$. It is said to be \emph{unfit} otherwise.
Moreover, a fit node is said to be \emph{reachable} if it is in the connected component of fit nodes that contains the root (as in Property F3). Equivalently, a node is reachable if either it is the root, or it is fit and its parent is reachable.
\end{definition}
Note that by definition, all nodes at depth $<\ldepth_1$ are fit and reachable.
The notion of fitness clearly satisfies the first fitness property F1.
We want to show that it also satisfies the properties F2 and F3 with $f(D, \delta) = (\delta^{-1} D)^{O(1/\eps)}$.
Let us first give two useful conditions satisfied by our choice of $\ldepth_1$ and $\ldepth_2$.
\begin{claim}\label{claim:technical}
The following inequalities hold:
\begin{enumerate}
\item \label{eq:first}
$2^{\ldepth_2} \Delta^{-\varepsilon \ldepth_1} \leq \frac{\delta}{4D}$,
\item \label{eq:second}
$2^{\ldepth_2}\Delta^{(1+\varepsilon)\ldepth_1 - \eps \ldepth_2}  \leq \frac{\delta}{4D}.$
\end{enumerate}
\end{claim}
\begin{proof}
{\emph{Equation \eqref{eq:first}}}. Replacing $\ldepth_{1/2}$ by their values, we bound the left hand side in Equation~\eqref{eq:first} as follows
\begin{align*}
2^{\frac{6}{\eps^2} \log_{\Delta} (4\delta^{-1} D )} &\Delta^{-\varepsilon \frac{2}{\eps} \log_\Delta (4\delta^{-1}D)}= 2^{(6\eps^{-2} - 2\log \Delta )\log_\Delta (4\delta^{-1}D)}.
\end{align*}
 We assume that $\Delta \geq 2^{6\eps^{-2}}$ so that $6\eps^{-2} \leq \log \Delta$ and $6\eps^{-2} -2 \log \Delta \leq -\log \Delta$. Hence the left hand side in Equation \eqref{eq:first} is not greater than
$2^{-\log \Delta \log_\Delta (4\delta^{-1}D)}= \frac {\delta}{4D}$.\\

\noindent {\emph{Equation \eqref{eq:second}}}. Using the fact that $\eps \leq 1$ and that $\ldepth_2 = \frac 3 \eps \ldepth_1$, we obtain
that \[(1+\eps)\ldepth_1 - \eps \ldepth_2 \leq 2 \ldepth_1 - 3\ldepth_1 = -\ldepth_1\leq - \eps \ldepth_1.\]
The result follows from Equation \eqref{eq:first}.
\end{proof}
\begin{proposition}\label{pro:succeed}
The notion of advice-fitness introduced above satisfies Property F2, namely, that for every node $u$ on the path $[\sigma, \tau]$,
we have $\Pr(u \mbox{ is fit})\geq 1-\frac{\delta}{2D}$.
Hence, using Claim \ref{claim:f2}, with probability at least $1- \frac{\delta}{2}$,  algorithm $A'_{walk}$ succeeds in finding $\tau$.
\end{proposition}
\begin{proof}
Let $u \in [\sigma, \tau]$. We want to upper bound the probability that $u$ is unfit.
The probability to be fit decreases with the distance to $\sigma$ until reaching depth $\ldepth_2$, by definition. Hence, it suffices to check the case where $u$ is at distance at least $\ldepth_2$ from $\sigma$.

Let $S$ be a set of $\ldepth_1$ nodes on $[u_{-\ldepth_2}, u]$.
The set $S$ is completely faulty with probability $q^{\ldepth_1} \leq \Delta^{-\varepsilon h_1}$.
The number of such sets S is at most $2^{\ldepth_2}$.
Hence $
\Pr(u\mbox{ is unfit}) \leq 2^{\ldepth_2} \Delta^{-\varepsilon \ldepth_1}.$
We conclude using Equation \eqref{eq:first} from Claim \ref{claim:technical}, that $\Pr(u\mbox{ is unfit}) \leq \frac{\delta}{4D} \leq  \frac{\delta}{2D}$.
\end{proof}

\begin{figure}[!ht]
    \centering
	    \includegraphics[width=0.45\textwidth]{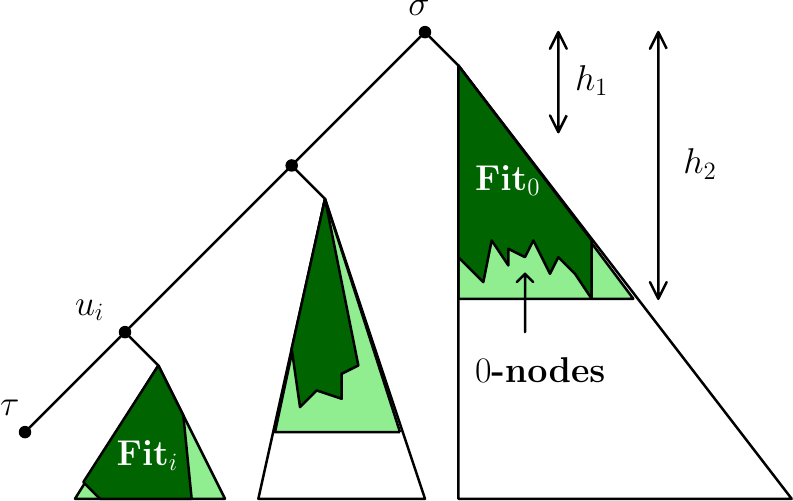}
        \caption{\small The partition of fit vertices introduced in the proof of Proposition \ref{pro:runtime}. The colored nodes in the subtree on the right are the close $0$-nodes, where those colored with dark green are the reachable fit $0$-nodes.  There are no fit $0$-nodes at depth greater than $h_2$ in this example.}\label{fig:fit}
\end{figure}

\begin{proposition}\label{pro:runtime}
The notion of advice-fitness obeys Property F3 with $f(D, \delta) = (\delta^{-1} D)^{O(\eps^{-1})}$. Hence
the move complexity of $A'_{walk}$ is less than $(\delta ^{-1}D)^{O(\eps^{-1})}$, with probability $\geq 1-\frac{\delta}{2}$.
\end{proposition}
\begin{proof}
Let $\mbox{Fit}$ be the connected set of reachable nodes (as defined in Definition \ref{def:fit}).
Our goal is to show that with high probability, namely, with probability at least  $1-\frac{\delta}{2}$, we have  $|\mbox{Fit}|=(\delta^{-1}D)^{O(\eps^{-1})}$.

For $i \ge 0$, the term {\em $i$-node} will refer to any node whose common ancestor with $\tau$ is at depth~$i$. An $i$-node is further said to be \emph{close} if its depth lies in the range $[i, i+ \ldepth_2]$.
Let $\mbox{Fit}_i$ be the set of close $i$-nodes in $\mbox{Fit}$ (see Figure \ref{fig:fit}).

Our first goal is to show that with high probability, Fit does not contain any $0$-node at depth $\ldepth_2$ (Claim \ref{claim:disco}).
Under this high probability event, $A'_{walk}$ visits only fit $0$-nodes that are close (i.e., at depth at most $\ldepth_2$), because $A'_{walk}$ visits only reachable nodes, and fit $0$-nodes that are not close are disconnected from the root at depth $\ldepth_2$. Hence all the $0$-nodes visited by $A'_{walk}$ are in $\mbox{Fit}_0$.
By symmetry, a similar statement holds for each layer $i$, and the corresponding subset $\mbox{Fit}_i$. Thus, under a high probability event, the nodes visited by $A'_{walk}$ during its execution form a subset of  $\bigcup_{i=0}^D \mbox{Fit}_i$ (namely, $f(D, \delta) \leq  \lvert \bigcup_{i=0}^D \mbox{Fit}_i \rvert$).

Denote the expected number of fit 0-nodes at depth $h$ by \[\alpha_h:=\sum_{\mbox{$u$ is a $0$-node at depth $h$}} \Pr(u\mbox{ is fit}).\]
We have
\begin{align}\label{eq:unifu}
\E\left(\lvert \mbox{Fit}_0 \rvert\right) &= \sum_{\mbox{$u$ is a close $0$-node}}\Pr(\mbox{$u$ is fit})=  \sum_{h\leq \ldepth_2}\alpha_h.
\end{align}

\begin{claim}\label{claim:bigU}
\begin{enumerate}
\item \label{eq:(*)}
$\sum_{\ldepth\leq \ldepth_1} \alpha_\ldepth \leq 2\Delta^{h_1}$.
\item
For every $\ldepth_1 < \ldepth\leq \ldepth_2$, we have  $\alpha_h\leq \Delta^{\ldepth_1 (1+ \eps)} 2^{\ldepth} \Delta^{-\eps \ldepth}.$
\end{enumerate}
\end{claim}
\begin{proof}
Every node at depth at most $\ldepth_1$ is fit. There are at most  $2\Delta^{h_1}$ such nodes. Estimation \eqref{eq:(*)} follows.

We now show the second estimate.
Let $U_h$ be a node chosen uniformly at random among all $0$-nodes at depth $h$. Then $\Pr(U_h\mbox{ is fit}) = \sum_{\mbox{$u$ is a close $0$-node}} \Pr(U_h = u) \Pr(u \mbox{ is fit})$, and hence we may write:
\[
\alpha_h = (\Delta-1)^{h}\Pr(U_h\mbox{ is fit}).\]
Choosing $U_h$ randomly rather than arbitrarily is of crucial importance in the adaptive semi-adversarial variant, because the adversary might choose to direct all the faulty advice towards a specific node $u$. This could result in some terms in the sum (in the definition of $\alpha_h$) being much bigger than the average. So it is important to avoid bounding the average by the max.
We draw a uniform path $\sigma = U_0, U_1, \ldots, U_h$ of length $h$ from the root, in the component of 0-nodes. Consider a node $U_i$ on this path. With probability $q$, it is faulty. In this case, regardless of how the adversary could set its advice,
 the advice at $U_i$ points to $U_{i+1}$ with probability of at most $\frac{1}{\Delta-1}$ over the choice of $U_{i+1}$.

It follows that the number of ancestors of $U_h$ whose advice points to $U_h$ may be viewed as the sum of $h$ Bernoulli variables $B_i$ with parameter $\frac q {\Delta-1}$. Moreover, the previous argument means that
$\Pr(B_i = 1 \mid B_j, j<i) = \frac q {\Delta-1}  =\Pr(B_i = 1).$ The $B_i$ variables are thus uncorrelated and hence independent since they are Bernoullis.
The node $U_h$ is fit if at least $\ldepth - \ldepth_1$ ancestors point to it.
Thus, by a union bound over the ${\ldepth \choose \ldepth - \ldepth_1} \leq 2^{\ldepth}$ possible locations of faults, $\Pr(U_h\mbox{ is fit}) \leq  2^{\ldepth} \left(\frac q {\Delta-1}\right)^{\ldepth - \ldepth_1}$. Hence
\[
\alpha_h \leq 2^{\ldepth}(\Delta-1)^{\ldepth} \left(\frac q {\Delta-1}\right)^{\ldepth - \ldepth_1} \leq 2^{\ldepth} q^{\ldepth- \ldepth_1} \Delta^{\ldepth_1} = \Delta^{\ldepth_1 (1+ \eps)} 2^{\ldepth} \Delta^{-\eps \ldepth}.
\]
In the last step, we used $q = \Delta^{-\eps}$. This concludes the proof of Claim \ref{claim:bigU}.
\end{proof}
For $\ldepth=\ldepth_2$,  combining Claim \ref{claim:bigU} and Equation \eqref{eq:second} stated in Claim \ref{claim:technical} implies:
\begin{equation}\label{eq:alpha}
\alpha_{h_2}\leq \frac{\delta}{4D}
\end{equation}
For $i \in [D]$, denote by $Z_i$ ($Z$ for {\em zero}) the event that there are no fit $i$-nodes at depth $i+\ldepth_2$.
Applying Markov inequality on Equation \eqref{eq:alpha} implies that $\Pr(Z_0)\geq 1- \delta (4D)^{-1}$. Since the same argument can be applied to any $i\leq D$, we get
\begin{claim}\label{claim:disco}
For every $i\leq D$, we have $\Pr(Z_i)\geq 1- \frac{\delta}{4D}$ and hence $\Pr(\bigcap Z_i) \geq 1- \frac{\delta}{4}$.
\end{claim}
If follows from the assumption on $\Delta \geq 2^{6 \eps^{-2}}$ that $\Delta^\eps \geq 2^{6} \geq 8$, so that $2\Delta^{-\eps} < \frac 14$.
Using Eq. \eqref{eq:unifu}, Claim \ref{claim:bigU} and the definition of $h_1$ we get:
\begin{align*}
\E\left(\lvert \mbox{Fit}_0 \rvert\right) &\leq \sum_{h\leq \ldepth_2}\alpha_h\leq 2\Delta^{h_1} + \sum_{\ldepth=\ldepth_1}^{\ldepth_2} \Delta^{\ldepth_1 (1+ \eps)} \left(2\Delta^{-\eps}\right)^h\\
&\leq 2\Delta^{h_1} + \Delta^{\ldepth_1 (1+ \eps)} \sum_{\ldepth \geq \ldepth_1} 4^{-h} \leq 2\Delta^{h_1} + 2
\Delta^{\ldepth_1 (1+ \eps)}\leq 4
\Delta^{\ldepth_1 (1+ \eps)}\\
&  \leq 4\cdot(4\delta^{-1} D)^{2\eps^{-1} + 2} \leq (4\delta^{-1} D)^{2\eps^{-1} + 3}.
\end{align*}
The computation is the same for every $i\leq D$, yielding that $\E(\lvert \mbox{Fit}_i|)\leq  (4\delta^{-1} D)^{2\eps^{-1} + 3}$, and by linearity of expectation, we obtain \[\E(|\bigcup_{i=0}^D \mbox{Fit}_i|)\leq D \cdot  (4\delta^{-1} D)^{2\eps^{-1} + 3}.\] Using the Markov inequality, with probability at least $1-\frac \delta{4D}$, this variable is upper bounded by
$4\delta^{-1} D^2 \cdot (4\delta^{-1} D)^{2\eps^{-1} + 3} \le (4\delta^{-1} D)^{2\eps^{-1} + 5}$.
As explained in the beginning of the proof, under the event $\bigcap_{i\leq D} Z_i$ (which occurs with probability at least $1-\frac{\delta}{4}$ thanks to Claim \ref{claim:disco}),
we have  $\mbox{Fit}\subset \bigcup_{i\leq D} \mbox{Fit}_i$. Using
a union bound, we conclude that with probability at least  $1-\frac{\delta}{2}$, we have
$|\mbox{Fit}| \le (4\delta^{-1}D)^{2\eps^{-1} + 5} = (\delta^{-1}D)^{O(\eps^{-1})}$, as desired.
\end{proof}

Claim~\ref{claim:meta} in combination with Propositions \ref{pro:succeed} and \ref{pro:runtime} proves Theorem~\ref{thm:walk-hp}.

\section{Lower Bound for High Probability Algorithms}\label{sec:lower-hp}
The goal of this section is to prove Theorem \ref{thm:lower-eps}. Informally, the theorem shows that the upper bound in Theorem \ref{thm:walk-hp} is the right bound, up to a constant factor in the exponent.

The proof is done for the query complexity, in a complete $\Delta$-ary tree of depth $D = \log_{\Delta} n$. This also implies at least the same lower bound for the move complexity.
Throughout the proof, $T$ denotes a complete $\Delta$-ary tree of depth $D$.
In our lower bound, the adversary places $\tau$ at a leaf of $T$ chosen at random according to the uniform distribution.
We denote by $F$ the set of faulty nodes (without directional advice).
Since this set as well as $\tau$ are chosen uniformly at random, we may assume without loss of generality that the algorithm is deterministic.
The presentation of our proof is simplified if we assume that the algorithm is told which nodes of $T$ are faulty (namely, are in $F$). We can make this assumption because it only strengthens our lower bound.

We reserve the letter $\subT$ to denote a subtree of $T$,  containing the descendants of its root (with respect to the original root $\sigma$). A subset $S$ is said to be \emph{completely faulty} if all nodes in $S$ are faulty. Overloading this expression, we say that a leaf $v$ of some subtree $\subT$ is \emph{completely faulty} if the path
from the root  of $\subT$ towards $v$ is completely faulty. The relevant reference subtree $\subT$ will be specified if it is not clear from the context. The number of completely faulty leaves of a subtree $H$ is denoted $B(\subT)$ or simply $B$ if $\subT$ is clear from the context.
When considering a subset $S \subseteq T$, we write $S^*$ to denote the pair $(S, S\cap F)$. In words, this corresponds to a subset with the information of which nodes are faulty.

\subsection{Proof of Theorem \ref{thm:lower-eps}}
At a high level, the argument is as follows. On the path from the root to $\tau$, with probability at least $\delta$, there exists a segment $[v_i,v_{i+h-1}]$ of length $h \simeq \frac 1 \varepsilon \log_{\Delta} (\frac D \delta)$, where all nodes are faulty (Lemma \ref{lem:faulty}). On the other hand, the tree rooted at $v_i$ of depth $h$,
 typically hosts many such completely faulty leaves (Lemma \ref{lem:manybis}). In some sense these leaves are \emph{indistinguishable}. This means that any algorithm has constant probability of trying a constant fraction of them before finding $v_{i+h}$, the one leading to $\tau$.

Let us make this intuition more precise. For each choice of faulty locations $F$ and treasure location $\tau = v$, we define $u(v,F)$ to be the first node on the path $[\sigma, v]$
such that $u$ and its $h-1$ ``direct'' descendants towards $v$, i.e., the next $h-1$ nodes on the path to $v$, are faulty, if such $u$ exists, and otherwise we say $u(v,F)$ is \emph{not defined}.

Denote by $\textbf{H}(v,F)$ the subtree rooted at $u(v,F)$ of depth $h$. A central object in the proof is $\textbf{H}^*(v,F)$ which corresponds to the couple $(\textbf{H}(v,F),(F\cap \textbf{H}(v,F)))$ (the subtree together with the faulty locations on it).
Henceforth, we will often drop the dependency on $v,F$ in the interest of readability, but we keep the bold notation to emphasize that
\textbf{H} is a random object (it depends on $F$). If $u(v,F)$ is not defined, we also say that $\specT$ is not defined.

The following lemma
lower bounds the probability that $\specT$ is well defined.
\begin{lemma}\label{lem:faulty}
 Let $0<\delta < \frac{1}{16}$. If $h$ satisfies $q^h \ge \frac{8\delta h}{D}$ and $D \ge \max[h,8\delta h]$ (for $D$ that does not satisfy this condition the statement does not make sense), then $\Pr\left(\mbox{$\specT$ is well defined}\right) \geq 4\delta.$
\end{lemma}

\begin{proof}
Recall that we write $[\sigma, \tau] := \{v_0 = \sigma, v_1, \ldots, v_{D-1}, v_D = \tau\}$.
For a given $i,h \in \mathbb{N}$,
let us denote by $F_{i,h}$ the event that $[v_i, v_{i+h-1}]$ is completely faulty.

With this definition, $\specT$ is well defined if the event $F_{i,h}$ holds for at least one value of $i$ in the range $0 \le i \le D-h$.  Hence what we want to show is that
\begin{equation}
\Pr\left(\mbox{NOT }\bigcup_{i=0}^{D-h} F_{i,h}\right) \leq 1-4\delta.
\end{equation}
For every fixed $i$ and $h$, $\Pr(F_{i,h}) = q^{h}$. Indeed, there are $h$ nodes on the path $[v_i, v_{i+h-1}]$ and each is independently faulty with probability $q$.
The parameter $h$ satisfies $q^{h} \ge \frac {8\delta h}{D}$ by assumption.

The events $F_{i'\cdot h,h}$ are independent when $i'$ varies in $[D/h]$. The probability that none of these holds is
\begin{align*}
(1-q^{h})^{D/h} \le \left(1-\frac {8\delta h}{D}\right)^{\frac{D}{h}} \leq  e^{-\frac {8\delta h}{D} \frac D h} = e^{-8\delta} \leq 1-8\delta/2 = 1-4\delta.
\end{align*}
The last inequality holds for sufficiently small $\delta$ (e.g., $\delta \le \frac{1}{16}$).
\end{proof}

From now on, we assume for simplicity that $h$ is chosen so that $q^h= \frac{8\delta h}{D}$. ($h$ being an integer, this is only an approximate equality in general. We ignore this point in the discussion, assuming $h$ has been appropriately rounded.)
Recall that $q=\Delta^{-\varepsilon}$, so taking logarithms we see that
$\eps h = \log_\Delta \frac{D}{\delta} - \log_{\Delta} 8h$. Viewing $\delta$ and $\Delta$ as fixed and letting $D$ go to infinity, the previous equality entails that
$h = \frac{1}{\eps}  (1-o_D(1)) \log_\Delta \frac{D}{\delta}$, where the term $o_D(\cdot)$ tends to $0$ when $D$ tends to infinity\footnote{
Indeed, $h$ tends to infinity when $D$ tends to infinity, so $\log_\Delta 8h = o(h)$.
}.

Let $\compleaves(\specT^*)$ denote the number of completely faulty leaves in $\specT^*=\textbf{H}^*(v,F)$.
Lemma \ref{lem:manybis} below is proven in  Section \ref{sec:manybis}. It bounds the typical value of $\compleaves(\specT^*)$ in those cases that $\specT$ is well defined.
\begin{lemma}\label{lem:manybis}
Let $C$ be a sufficiently large constant that depends only on $q$. Then
\[
\Pr\left(\compleaves(\specT^*) \leq (q\Delta)^{\ldepth-C} \mid \specT~ \mbox{is well defined}\right) \leq 0.5.
\]
\end{lemma}

The following two intermediate results express how ``indistinguishable'' is formalized in this context.
\begin{claim}\label{claim:indist}
Consider a leaf $v$ and a subtree $\specT^*$ (together with the faulty locations in it). Then, the value
of $p_{v,H^*}:=\Pr_F(\specT^*(v,F) = H^* \mid \tau = v)$ is the same for all $v$ such that $p_{v,H^*}>0$.
\end{claim}
\begin{proof}
Let $H$ be a fixed subtree of depth $h$ rooted at some node $u$. By definition, the statement that $\Pr(\specT^*(v,F) = H^* \mid \tau=v) > 0$ is equivalent to the following two statements:
\begin{itemize}
\item[$\mathcal{A}$:] There are no $h$ consecutive faulty nodes in $[\sigma,u]$ and,
if $u \neq \sigma$,
then $u$'s parent is not faulty.
\item[$\mathcal{B}$:] Leaf $v$ is a descendant of $u$ and
the leaf of $H^*$ which is an ancestor of $v$ is completely faulty in $H^*$.
\end{itemize}
The probability of $\mathcal{A}$ depends only (in some complicated way) on the length of $[\sigma, u]$ and $q$ and hence does not depend on $v$.
With these notations, if $v$,$H^*$ and $F$ are such that
$\mathcal{B}$ holds, then
$$
\Pr(\specT^*(v,F) = H^* \mid \tau=v) = q^{\lvert F \cap \subT \rvert} (1-q)^{\lvert H \setminus F\rvert}\Pr(\mathcal{A}).
$$
The right hand side does not depend on $v$.
The claim follows.
\end{proof}

\begin{lemma}\label{lem:uniform}
Conditioning on the subtree $\specT^*$ (and hence its existence), the leaf of $\specT^*$ which leads to the treasure is  uniform amongst all completely faulty leaves $v$ of $\specT^{*}$.
\end{lemma}
\begin{proof}
Denote by $\mathcal{L}$ the set of leaves of $T$. Using Bayes rule, and because we assume that $\tau$ is uniform over all leaves $\mathcal{L}$,
\begin{align*}
\Pr(\tau=v \mid \specT^*) &= \Pr(\tau=v) \cdot \frac{\Pr(\specT^* \mid \tau=v)}{\Pr(\specT^*)}= \frac{1}{\lvert \mathcal{L}\rvert}\cdot \frac{\Pr(\specT^* \mid \tau=v)}{\Pr(\specT^*)}.
\end{align*}
It follows that
$\Pr(\tau=v \mid \specT^*)$ has the same value for all leaves $v$ of $T$ such that $\Pr(\tau=w \mid \specT^*) > 0$.
Indeed, we saw that the right hand term is independent of $w$, as soon as $w$  is a descendant of a completely faulty leaf of $\specT^*$ (Claim \ref{claim:indist}), and otherwise it is~$0$.

Since the tree $T$ is complete and regular, each leaf of $\specT^*$ is the ancestor of the same number of leaves in $T$, and so each completely faulty leaf of $\specT^*$ is equally likely to lead to the treasure.
\end{proof}

We now condition on the event that $\specT$ is well defined and that it
has more than $s = (q\Delta)^{\ldepth - C}$ completely faulty leaves.
This event holds with probability at least $4\delta \times 0.5$ (combining the results of Lemma \ref{lem:faulty} and Lemma \ref{lem:manybis}).
Under this conditioning, with probability at least $0.5$ over treasure location the completely faulty leaf leading to the treasure is visited after at least $0.5s$ other faulty leaves have been visited.
Indeed there are $s$ faulty leaves, each being equally likely to lead to the treasure (Lemma \ref{lem:uniform}).
Overall, with probability $4\delta \cdot 0.5 \cdot 0.5 = \delta$, more than $0.5 s$ nodes need to be visited.
We saw that, $h = \frac{1}{\eps}  (1-o_D(1))\log_\Delta (\frac{D}{\delta}){}  $, hence
$s = (q\Delta)^{h-C} =  (q\Delta)^{h \cdot (1-o_D(1))}= (\Delta^{1-\eps})^{\frac{1}{\eps} (1-o_D(1)) \log_\Delta \frac{D}{\delta} }$. After simplification, this is
$(\delta^{-1} D)^{\frac {1-\eps} \eps  (1-o_D(1))}$.

\subsection{Proof of Lemma \ref{lem:manybis}}\label{sec:manybis}

The proof of Lemma \ref{lem:manybis} is broken into intermediate claims.
To begin with we ignore the treasure $\tau$, and consider a fixed complete $\Delta$-ary tree $\subT$ of depth $h$, with root $\sigma$. Each node of $\subT$ is faulty (namely, belongs to the set $F$) independently with probability $q$. Let $\compleaves$ denote the number of completely faulty leaves in $\subT$.
\begin{claim}\label{claim:branching}
It holds that
$
\Pr\left(\compleaves>\frac 12 (q\Delta)^{h}\right) = \Omega(q).$
\end{claim}
\begin{proof}
The proof uses a second moment argument. For every given leaf, the probability of the full path from the root being faulty is $q^{h}$ and there are $\Delta^h$ leaves.
Hence $\E(\compleaves) = (q\Delta)^h$.
Let us denote by $\mathcal{L}$ the set of leaves.
Using the Boolean indicator variable $\chi_i$ to denote that leaf $i$ is completely faulty, we have $B=\sum_i \chi_i$. Hence, we may write $\E(\compleaves^2)=\E((\sum_i \chi_i)^2)$ as
\[
\E(\compleaves^2) = \E(\compleaves) + \sum_{u \neq v \in \mathcal{L}} \Pr(\mbox{$u$ and $v$ are completely faulty}),\]
where the sum is taken on the ordered pairs $u \neq v$ where both are  in $\mathcal{L}$. Fix a leaf $v\in \mathcal{L}$, and let $\mathcal{L}_\ell$ denote the set of leaves whose common ancestor with $v$ has depth
$h-\ell$. For every $\ell \in [1, h]$, $\lvert \mathcal{L}_\ell \rvert \leq \Delta^\ell$.
Moreover, for $u \in \mathcal{L}_\ell$,
$u$ and $v$ being completely faulty is equivalent to $v$ being completely faulty and the $\ell-1$ nodes connecting $u$ to the root-to-$v$ path being faulty. Hence,
$\Pr(\mbox{$u$ and $v$ are completely faulty}) = q^{\ell + h -1}$.
Altogether,
\begin{align*}
\E(\compleaves^2) \leq (q\Delta)^h + \Delta^h \sum_{\ell=1}^h\Delta^\ell q^{\ell + h -1}
&= O\left(q^{-1}(q\Delta)^h \sum_{\ell=1}^h (q\Delta)^\ell \right)\\
&= O\left(q^{-1}(q\Delta)^{2h}\right) = O(q^{-1}\E(\compleaves)^2).
\end{align*}
Using the Paley-Zygmund inequality, we get
$
\Pr(\compleaves \geq \frac 12 \E(\compleaves)) \geq \frac 14 \frac{\E(\compleaves)^2}{\E(\compleaves^2)} = \Omega(q).
$
\end{proof}

\begin{claim}\label{claim:many}
For a constant $C$ that depends only on $q$, 
$
\Pr(\compleaves \leq (q\Delta)^{h-\smallfun} \mid \compleaves \geq 1) \leq 0.5.
$
\end{claim}
\begin{proof}
First observe that for any constant $\smallfun$, we may consider only $h\geq \smallfun$, since otherwise, if $h< \smallfun$, then $(q\Delta)^{h-\smallfun}<1$ and the requirement trivially holds. 

Since $\compleaves\geq 1$ there exists a path $[\sigma, v]$ which is completely faulty.
For every $u \in [\sigma, v]$ define $T_u$ as the subtree rooted at $u$ excluding the subtree rooted at the child of $u$ on $[\sigma, v]$. The subtrees $T_u$ are pairwise disjoint and form a partition of $T$.
For $u \in [\sigma, v]$, define $\compleaves_u:=\compleaves(T_u)$, the number of completely faulty leaves of $T_u$.
Note that $\compleaves =\sum_u \compleaves_u \geq \max_u \compleaves_u$.
Moreover since, for $u\neq u'$, $T_u \cap T_{u'} = \emptyset$, the variables $\compleaves_u$ are independent.

Let $S$ be the prefix of size $\smallfun$ of $[\sigma, v]$.
All subtrees rooted at a node $u\in S$ are of depth $> \ldepth-\smallfun$. 
Since $q\Delta>\Delta^{1/2}$, we have $(q\Delta)^{-\smallfun}<1/2$ for every $\smallfun\geq 2$.
Using Claim \ref{claim:branching}, together with the independence of the $\compleaves_u$'s, the probability that all $\compleaves_u$'s are smaller than $(q\Delta)^{h-\smallfun}\leq \frac{1}{2}(q\Delta)^{h}$ is less than
$(1-cq)^{\smallfun}$ for a  constant $c$.
 The result follows, if $C$ is large enough (as a function of  $q$).
\end{proof}

If it exists, by definition, $\specT^*$ has at least one completely faulty leaf, which is the one leading to $\tau$. Outside of the branch leading to $\tau$, the nodes of $\specT^*$ are still independently faulty with probability $q$. This means that
the number of completely faulty leaves of $\specT^*$, $\compleaves(\specT^*) \mid \{ \specT^* \mbox{ is well defined}\}$ is distributed as $\compleaves(H^*) \mid \{\compleaves(H^*) \geq 1\}$ for every fixed subtree $H$ of depth $h$.

Using this together with Claim \ref{claim:many} finishes the proof of Lemma \ref{lem:manybis}.
Indeed, we obtain
$$
\Pr\left(\compleaves(\specT^*) \leq (q\Delta)^{h-\smallfun} \mid \specT~\mbox{is well defined}\right) \leq 0.5.
$$


\section{A Query Algorithm for the Path}\label{sec:path}
This section and the following two sections are devoted to the  analysis of upper bounds on the query complexity.
In the current section we begin this analysis by  focusing on
 the special case that the tree is a path. For the path, we can adapt to our setting of permanent faults algorithms that were developed in models in which faults are not permanent, and repeated queries to the same node results in independent replies. Specifically, the path algorithm that we present and its analysis can be thought of as an adaptation of an algorithm of~\cite{bayes08} to our setting.

Our query model for paths can conveniently be described as follows. The path consists of $n$ vertices $\{1, \dots, n\}$, one of which (chosen by an adversary) contains a treasure. Denote this node by $\tau$. Every node other than $\tau$ holds advice, chosen independently at random, pointing towards $\tau$ with probability $p > \frac{1}{2}$ and away from $\tau$ with probability $1-p$. A query to a node reveals its advice (and if no advice is present, the node must be $\tau$). We propose the following algorithm to find the treasure with a small expected number of queries.

\paragraph{The median algorithm.}
Initially, every node has weight of $1$ and all nodes are {\em live}. In every step do the following.

\begin{itemize}

\item
Query the live node $i$ such that the weight on each side of it is no more than half the total weight. We refer to this node as the {\em median}. (There might be two such nodes, one with exactly half the weight below it and the other with exactly half the weight above it. In this case, query one of them arbitrarily.)

\item
If $i = \tau$, the algorithm ends.

\item
If $i \not= \tau$ then $i$ is declared dead and its weight is dropped to~0. If the advice at $i$ points up, multiply the weight of each node above $i$ by $2p$, and the weight of each node below $i$ by $2(1-p)$.  If the advice at $i$ points down, multiply the weight of each node below $i$ by $2p$, and the weight of each node above $i$ by $2(1-p)$.

\end{itemize}
\medskip
We now provide some intuition regarding the number of queries made by the median algorithm.
Let $H(p) = -p\log p - (1-p)\log (1-p)$ (all logarithms are in base~2). Hence $H$ is the entropy of $p$. Observe that for $\frac{1}{2} < p \le 1$ we have that $1 > H(p) \ge 0$. We further let $I(p) = 1 - H(p)$. This can be viewed as the information content of advice. Had there been no faults, the advice would contain one bit of reliable information, but given the probability of fault, the information content decreases by the entropy of $p$.

Given the fact that there are $n$ possible location for the treasure and each query gives $I(p)$ amount of information (in fact, slightly more, because it also excludes the queried node from containing the treasure), we do not expect to find the treasure in fewer than $\frac{\log n}{I(p)}$ queries (up to low order terms). The median algorithm manages to ask the most informative queries, and indeed finds the treasure in roughly $\frac{\log n}{I(p)}$ steps in expectation (up to low order terms), where expectation is taken over the choice of advice. The proof of the next theorem is sketched in Appendix~\ref{sec:median}.

\begin{theorem}
\label{thm:median}
The expected number of queries of the median algorithm is $\frac{\log n}{I(p)}$, up to low order terms.
\end{theorem}

Let us now try to extend the median algorithm to the case that the tree $T$ is not a path. This requires a notion of a median vertex in a tree. Fortunately, such a notion exists, and is referred to as a {\em separator} vertex. Given arbitrary nonnegative weights for the vertices of a tree $T$, we call $u$ a \emph{separator} of $T$ if each connected component of $T\setminus \left\{u\right\}$ contains at most half the weight of the vertices. If is well known that for every tree and every weight function, a separator node exists. However, this by itself does not suffice in order to extend the median algorithm to trees. The difficulty is that at various steps of the algorithm, the respective median node might be dead, and hence no information can be inferred by querying it again (in our model where faults are permanent). This problem does not occur on the path because given a nonnegative weight function for its vertices, there always is a median node of strictly positive weight (hence, one that was not previously queried by the algorithm), whereas on trees the median node might be unique and of weight~0 (and hence dead).

Circumventing the above problem is not easy. Inevitably, our algorithms spend most of their queries on non-separator vertices.
Unlike the median algorithm for the path whose number of queries is optimal up to low order terms, the algorithms that we shall design for trees will use a number of queries that is only approximately optimal.
Our approach is easiest to explain in the high probability setting, and this we do in Section~\ref{sec:corquery}. Then, for the expectation setting, we provide one algorithm in Section~\ref{sec:weak-upper}, and a refined algorithm with a stronger upper bound on the expected number of queries in Appendix~\ref{sec:regular-upper}. A lower bound (that does not quite match the upper bounds) is provided in Section~\ref{sec:lb-query}.

\section{A High Probability Query Upper Bound}\label{sec:corquery}
In this section we prove Theorem \ref{cor:query}.
The idea is based on a separator search.
We call $u$ a \emph{separator} vertex of tree $T$ (here, unlike the previous section, all vertices in $T$ are assumed to have equal weight) if each connected component of $T\setminus \left\{u\right\}$ contains at most $|T|/2$ nodes. It is well known that such a node exists.

We use a local procedure described in Lemma \ref{lem:local-query} that allows us to learn  with probability $1 - O(\frac{\delta}{\log n})$, in which one of the connected components of $T \setminus \{u\}$ the treasure resides. By a union bound, applying this local procedure on a separator of the tree, and recursing on the component  pointed out by the procedure, allows to find the treasure in logarithmic number of runs of the local procedure with probability at least $1 - O(\delta)$.

\begin{lemma}\label{lem:local-query}
Consider the adaptive semi-adversarial model. Let $d$ be any positive integer and let $u$ be a separator. There exists a search procedure that queries the vicinity of $u$ up to distance $d$ and outputs either $(*)$ the component of $T\setminus \{u \}$ that contains $\tau$ or $(**)$ $\tau$ itself if $d(\tau, u) \leq d$. The success probability is at least $1- \delta$ and the number of queries is not greater than $(\delta^{-1} d)^{O(\eps^{-1})}$.
\end{lemma}

\begin{proof}
Let $T_u$ be the subtree of depth $d$ rooted at $u$.
Hereafter, the reference tree is $T_u$, so the notion of fitness is with respect to $T_u$ (that is, the depth parameter involved in Definition~\ref{def:fit} is $d$ and not $D$).

Let us first describe the promised local procedure. It consists in applying algorithm $A'_{walk}$ from Theorem \ref{thm:walk-hp} on $T_u$ until either finding the treasure or finding a reachable fit node at distance precisely $d$ from $u$, denoted $x$. We will see that with high probability at least one of these events hold so that the behavior of the local procedure when none of these events happens is not relevant. For the sake of concreteness, we could say that it stops if all of $T_u$ has been explored.
We note that Algorithm $A'_{walk}$ makes walking steps, which are viewed as queries in the query model in this context. The output of the local search procedure is either $\tau$, if it was found, or the component of $x$ in $T\setminus \{u \}$.

Let us analyse the performance of the local search procedure.
If $\tau \in T_u$, then Theorem~\ref{thm:walk-hp} ensures that it is found with probability $1-\delta$ in at most $(\delta^{-1} d)^{O(\eps^{-1})}$ steps.
Otherwise, consider the node $x$ at distance $d$ from $u$ in $T_u$ that is on the path to $\tau$. Within $T_u$, the advice is sampled as if  $x$ was the treasure $\tau$, so Theorem~\ref{thm:walk-hp} guarantees in this case that $x$ is found with probability $1-\delta$ in less than $(\delta^{-1} d)^{O(\eps^{-1})}$ steps.
Moreover, under that event $x$ is fit and reachable.

To complete the argument, we just need to guarantee that with high probability, there are no reachable fit nodes at distance $d$ from $u$ outside the component of $x$.
Then, the local procedure may discover a reachable and fit node at distance $d$ different from $x$, but it will still be in the same component as $x$.

Recall that a $0$-node is a node whose common ancestor with the treasure is the root (which is $u$ in this case, since the reference tree is~$T_u$). These are precisely the nodes in the other components than the component of $x$. Claim \ref{claim:disco} asserts that all reachable fit $0$-nodes are within distance $\ldepth_2(d)$ of $u$  with probability at least $1-\frac{\delta}{4}$. We write $\ldepth_2(d)$ to emphasize that the parameter is defined here as a function of $d$, namely
$\frac{6}{\eps^2}\log_{\Delta} (4\delta^{-1} d)$. We see that if $d$ is big enough (as a function of $\eps, \Delta$) then $\ldepth_2(d) < d$.

Hence with that probability, every fit node at distance $d$ found by the local procedure is guaranteed to be in the right component, that is the component to which $x$ belongs.

Overall, the success probability of this procedure is $1-\delta - \frac{\delta}{4} = 1-\frac 54 \delta$. We may write $\delta'= \frac 45 \delta$ and get the desired statement\footnote{In fact, the rescaling $\delta' = \frac 45 \delta$ could be avoided, by observing that the event of probability $1-\delta /4$ discussed in the proof is included in the $1-\delta$ probability event that guarantees the success of $A'_{walk}$. This follows from inspecting the proof of Theorem \ref{thm:walk-hp}.}.
\end{proof}
To conclude the proof of Theorem \ref{cor:query}, we apply Lemma \ref{lem:local-query} with $\delta' = \delta /\log n$ and $d=\log n$ for every one of the at most $\log n$ separators leading to $\tau$.
By a union bound, the success probability is $1-\delta$ and the total number of queries is
$\log n \cdot (\delta^{-1} \log n)^{O(\varepsilon^{-1})} = (\delta^{-1} \log n)^{O(\varepsilon^{-1})}$.

\section{An $O(\sqrt{\Delta} \log \Delta \log^2 n)$ Upper Bound for Query Algorithms in Expectation}\label{sec:weak-upper}

This section is devoted to the proof  of Theorem \ref{thm:weak-upper}.
This theorem states that  for every $\varepsilon>0$, there exists a  deterministic query algorithm $\algcontract$ such that if Condition \condition holds with parameter $\varepsilon$ (see Eq. \eqref{eq:condition}), then $\algcontract$ needs at most $\bigO(\sqrt{\Delta}\log \Delta \cdot \log^2 n)$ queries in expectation.

As in the previous section, our technique in this section is based on separators.
Assume there is some local procedure, that given a vertex $u$ decides with probability $1 - \delta$ in which one of the connected components of $T \setminus \set{u}$, the treasure resides.
Applying this procedure on a separator of the tree, and then focusing the search recursively only in the component it pointed out, results in a type of algorithm we call a \emph{separator based} algorithm.
It uses the local procedure at most
$\Ceil{\log_2 n}$ times, and by a union bound, finds the treasure with probability at least $1 - \Ceil{\log_2 n } \delta$. Broadly speaking, we will be interested in the expected running time of this sort of algorithm conditioned on it being successful. This sort of conditioning complicates matters slightly.

In the remaining of the section, we assume that the set of separators for the tree is fixed.
We also use the notation $\queries(A)$ to denote the expected number of queries an algorithm $A$ uses before finding the treasure.

\begin{proof} \emph{(of Theorem
\ref{thm:weak-upper})}

$\algsep$ runs a separator based algorithm in parallel (i.e., in an alternating fashion) to some arbitrary exhaustive search algorithm.
Fix some small $h$, to be specified later. The local exploration procedure, denoted $\local$, for a vertex $u$ proceeds as follows.\newline

\noindent
\textbf{Procedure $\local(u)$.}
Recall from Eq. \ref{eq:beta} that $\beta(u) = \prod_{w \in \pathco{\sigma}{u}} \Delta_w$.
Consider the tree $T_h(u)$ rooted at $u$ consisting of all vertices satisfying $\log_\Delta\beta(v) < h$ together with their children. So a leaf $v$ in $T_h(u)$ is either a leaf of $T$, or
satisfies $\Delta^h \leq \beta(v) < \Delta^{h+1}$. Denote the second kind a {\em nominee}.
Call a nominee {\em promising} if the number of weighted arrows pointing to $v$ is large, specifically, if
$
\sum_{w\in [u,v \rangle} X_w\geq \F 2 3  h \log\Delta
$,
where  $X_w = \log \Delta_{w}$ if the advice at $w$ is pointing to $v$, $X_w = - \log \Delta_{w}$ if it is pointing to $u$, and
$X_w = 0$ otherwise.
Viewing it as a query algorithm, we now run the walking algorithm $\algzplus$ on $T_h(u)$ (starting at its root $u$) until it either finds the treasure or finds a promising nominee.
In the latter case, $\local(u)$ declares that the treasure is in the connected component of $T \setminus \set{u}$ containing this nominee.
If $\tau \in T_h(u)$ then set $\leaf[u] = \tau$. Otherwise let $\leaf[u]$ be the leaf of $T_h(u)$ closest to the treasure, and so in this case $\leaf[u]$ is a nominee. Denote by $\mathcal{U}(u)$ the set of nominees that are not in the same component as $\leaf[u]$ in $T \setminus \{ u \}$.
Say that $u$ is \emph{$h$-misleading} if either:
\begin{itemize}
\item $\tau \not\in T_h(u)$ and $\leaf[u]$ is not promising, or
\item there is some promising nominee $v \in T_h(u)$ that is in  $\mathcal{U}(u)$.
\end{itemize}
In particular, if $u$ is not $h$-misleading then $\local(u)$ necessarily outputs the correct component of $T \setminus \set{u}$, namely, the one containing the treasure.
The proof of the following lemma is to be found below in Section \ref{sec:lemmisleading}. The part regarding uniform noise will be needed later. Recall we say the noise is uniform if it does not depend on the node so that $q_u = q$ for every node $u$.

\begin{lemma}\label{lem:misleading}
For every  $u$,
$\prob{u \text{ is $h$-misleading}} \leq (\Delta + 1)(1-\varepsilon)^h$.
Also, for every event $X$ such that $X$ occurring always implies that $u$ is not $h$-misleading, we have
\[\prob{X}\queries\B{\local(u) \mid X} = \bigO(\sqrt{\Delta} \log \Delta \cdot h).\]
In the case the noise is uniform these bounds become
$2(1-\varepsilon)^h$ and
$\bigO(\sqrt{\Delta} \cdot h)$
respectively.
The constant hidden in the $\bigO$ notation only depends polynomially on $1/\varepsilon$.
\end{lemma}
Applying Lemma \ref{lem:misleading}, with $h = - 3\log(2n) / \log(1 - \varepsilon)$, gives
$\prob{u \text{ is misleading}} \leq 1/n^2$.
Denote by $\good$ the event that none of the separators encountered are misleading, and by $\good^c$ the complement of this event.
By a union bound,
$\prob{\good^c} \leq 1/n$.
\begin{equation}\label{eq:decompose}
\queries(\algsep) =
\prob{\good}\queries\left(\algsep \mid \good \right)
+
\prob{\good^c}\queries\B{\algsep \mid \good^c}.
\end{equation}
As $\algsep$ runs an exhaustive search algorithm in parallel, the second term is $\bigO(1)$. For the first term, note that conditioning on $\good$, all local procedures either find the treasure or give the correct answer, and so there are $\bigO(\log n)$ of them and they eventually find the treasure.
Denote by $u_i$ the $i$-th vertex that $\local$
is executed on. By linearity of expectation, and applying Lemma \ref{lem:misleading}, the first term of \eqref{eq:decompose} is
$
\prob{\good}\sum_i \queries\B{\local(u_i) \mid \good}
= \bigO(\log n \cdot \sqrt\Delta \log \Delta  \cdot h)
= \bigO(\sqrt\Delta \log \Delta  \log^2 n).$
As $e^{-x}>1-x$ always, then  $-1/\log(1 - \varepsilon) \leq 1/\varepsilon$, and the hidden factor in the $\bigO$ is as stated.
This establishes Theorem \ref{thm:weak-upper}, conditioning on proving Lemma \ref{lem:misleading}.
\end{proof}

\subsection{Proof of Lemma \ref{lem:misleading}}\label{sec:lemmisleading}
The proof makes use of Lemma \ref{lem:maintec}
given in Section \ref{sec:walk-exp}.
To check the probability that $u$ is misleading, consider two cases:
\begin{enumerate}
\I
$\tau \not \in T_h(u)$, and $\leaf[u]$ is not promising.
By Lemma \ref{lem:maintec}, and recalling that $\Delta^h \leq \beta(\leaf[u])$, the probability $\leaf[u]$ is not promising is:
\eq{
&
\prob{\sum_{w\in \pathco{u}{\leaf[u]}} X_w \leq \F 2 3 h \cdot \log(\Delta)}
\\ & \leq
\prod_{w\in \pathco{u}{\leaf[u]}} \F{1-\varepsilon}{\sqrt{\Delta_w}} \cdot e^{\F 34 \cdot \F 2 3 h \log(\Delta)}
=
\F{(1-\varepsilon)^{d(u, \leaf[u])}}{\sqrt{\beta(\leaf[u])}} \Delta^\F{h}{2}
\leq
(1-\varepsilon)^{d(u, \leaf[u])}.
}

As $d(u, \leaf[u]) \geq \log_\Delta \beta(\leaf[u]) \geq h$, this is at most $(1-\varepsilon)^h$.
\I
If $v$ is a nominee that is not in the same connected component of $T\setminus \set{u}$ as $\leaf[u]$, then by Lemma \ref{lem:maintec}, the probability that $v$ is promising is
\begin{align*}
\PP\B{\sum_{w\in \pathco{u}{v}} X_w \geq \F 2 3 \log \Delta \cdot h}
&=
\PP\B{\sum_{w\in \pathco{u}{v}} -X_w \leq -\F 2 3 \log \Delta \cdot h}\\
&\leq
\prod_{w\in \pathco{u}{v}} \F{1-\varepsilon}{\sqrt{\Delta_w}} \cdot e^{-\F 34 \cdot \F 2 3 h \log \Delta}\\
&=
\F{(1-\varepsilon)^{d(u,v)}}{\sqrt{\beta(v)}} \Delta^{-\F{h}{2}}
\leq
\F{(1-\varepsilon)^{d(u,v)}}{\Delta^h}.
\end{align*}

However, denote by $L$ the set of nominees in $T_u$. As they are a subset of the leaves of $T_u$, by the way $\theta$ is defined:
\begin{equation}\label{eq:numLeaves}
1 \geq \sum_{x \in L} \theta(v) \geq \sum_{x \in L} \R{\beta(v)} \geq \sum_{x \in L} \R{\Delta^{h+1}} = \F{|L|}{\Delta^{h+1}}
\end{equation}
So, $|L| \leq \Delta^{h+1}$. Therefore, by a union bound, the probability that there exists a nominee $v$ that renders $u$ misleading is at most $\Delta (1 - \varepsilon)^h$.
\end{enumerate}
The probability that $u$ is misleading is then at most $(1 + \Delta)(1 - \varepsilon)^h$ as stated. In the case where the tree is regular, the analysis is the same, except that in \eqref{eq:numLeaves}, $\beta(v) = \Delta^h$, and so following the same logic, $|L| \leq \Delta^h$, and this part contributes only $(1-\varepsilon)^h$.

For the second part of the lemma, consider some event $X$ where $u$ is not misleading.
As $\leaf[u]$ is either the actual treasure or promising, and acts as the treasure in the eyes of $\algzplus$, then the local procedure stops when it encounters $\leaf[u]$. It might actually stop before (because it found another promising node),
so,
\begin{align*}
\prob{X}\queries\B{\texttt{local}(u) \mid X} &\leq
\prob{X}\queries\B{\algzplus(T_h(u)) \mid X}
\\
&\leq \queries\B{\algzplus(T_h(u))} =
\bigO(\sqrt\Delta \cdot \texttt{depth}(T_h(u)))
\end{align*}
But the depth of $T_u$ is at most $\bigO(h\log\Delta)$, since its leaves satisfy
$\beta(v) < \Delta^{h+1}$, and $\beta(v) \geq 2^{\texttt{depth}(v)}$.
For the case of a regular tree, $\beta(v) = \Delta^{\texttt{depth}(v)}$ and so the depth of $T_u$ is at most $h$, giving the result.

\section{A Query Lower Bound of $\Omega(\sqrt{\Delta} \cdot \log_{\Delta} n)$ when  $q\sim {1}/{\sqrt{\Delta}}$}\label{sec:lb-query}

We now prove Theorem \ref{thm:lower-exp}. Specifically, we wish to prove that for $\Delta \geq 3$, on the complete $\Delta$-ary tree of depth $D$,
every algorithm needs $\Omega(q \Delta D)$ queries in expectation. Note that, in particular, when $q$ is roughly $1/\sqrt{\Delta}$, and $n$ is the tree size, the query complexity becomes $\Omega(\sqrt{\Delta} \cdot \log_{\Delta} n)$.

To prove the lower bound of $\Omega(q \Delta D)$, consider the complete $\Delta$-ary tree of depth $D$.  We prove by induction on $D$, that if the treasure is placed uniformly at random in one of the leaves, then the expected query complexity of every algorithm is at least $q(\Delta/2 - 1)D$.
If $D=0$, then there is nothing to show.
Assume this is true for $D$, and we shall prove it for $D+1$.
Let $T_1,\ldots ,T_{\Delta-1}$ be the subtrees hanging down from the root (in the induction, the ``root'' is actually an internal node, and so has $\Delta-1$ children), each having depth $D$.
Let $i$ be the index such that $\treasure \in T_i$,
and denote by $Q$ the number of queries before the algorithm makes its first query in $T_i$.
We will assume that the algorithm gets the advice in the root for free.
Denote by $Y$ the event that the root is \faulty. In this case, Observation \ref{choosing} applies, and we need at least $\Delta/2 - 1$ queries to hit the correct tree.
We subtracted one query from the count because we want to count the number of queries strictly before querying inside $T_i$.
We therefore get
$\expct{Q} \geq \prob{Y} \cdot \cexpct{Q}{Y} \geq q(\Delta/2 - 1).$
By linearity of expectation, using the induction hypothesis, we get the result for a uniformly placed treasure over the leaves, and so it holds also in the adversarial case.
\qed

\section{Open Problems}\label{sec:open}
As mentioned, the model with permanent noise can be extended to general graphs. Essentially, when a node $u$ is correct its advice points to one of the neighbors of $u$ on a shortest path to the treasure. As there might be several such neighbors, one may consider an adversary that chooses which of these neighbors to point to. In the purely probabilistic setting, with probability $q$, each node is faulty, in which case its advice points to a random neighbor. The semi-adversarial setting can similarly be defined. Obtaining efficient search algorithms for general graphs is highly intriguing. Even though the likelihood
of a node being the treasure under a uniform prior can still be computed in principle, it is not clear how to compare two nodes as in
Theorem \ref{thm:zplus} because there may be more than a single path between them.

In a limited regime of noise, we believe that memoryless strategies might very well be efficient also on general graphs, and we pose the following conjecture.
Proving it may require the use of tools from the theory of Random Walks in Random Environments, which seem to be lacking in the context of general graph topologies.
\begin{conjecture}
There exists a probabilistic following algorithm that finds the treasure in expected linear time on every undirected graph assuming $q<{c}/{\Delta}$ for a small enough $c>0$.
\end{conjecture}



\bibliographystyle{plain}
\bibliography{biblio}

\clearpage
\centerline{\huge{Appendix}}
\appendix

\section{Analysis of the median algorithm}\label{sec:median}

In this section we prove Theorem~\ref{thm:median}, concerning the expected number of queries for the median algorithm on the path.

\begin{proof}
Observe that if the weight of $\tau$ exceeds half the total weight of all nodes, the algorithm queries $\tau$.

\begin{lemma}
In every step, the total weight of all nodes does not increase.
\end{lemma}

\begin{proof}
Let $W$ be the total weight of all nodes before the query to node $i$, let $W^+ \le \frac{W}{2}$ be the total weight of nodes above $i$ and let $W^-\le \frac{W}{2}$ be the total weight of nodes below $i$. If the advice in $i$ points up then the total weight becomes $2pW^+ + 2(1-p)W^- \le pW + (1-p)W \le W$. A similar bound holds if the advice points down.
\end{proof}

Let $L_t$ denote the logarithm (in base 2) of the total weight of all nodes after step $t$. Initially, $L_0 = \log \left(\sum_{i=1}^n 1\right) = \log n$. The lemma above implies:

\begin{corollary}
For every $t \ge 1$, $L_t \le \log n$.
\end{corollary}

Recall the notation of $H(p)$ and $I(p)$.
Let $\ell_t$ denote the logarithm (in base~2) of the weight of $\tau$ after step $t$. Observe that $\ell_0 = \log 1 = 0$. Let $E[\ell_t]$ denote the expectation of this random variable, where expectation is taken over choice of random advice.

\begin{lemma}
In every step $t$:
$$E[\ell_t - \ell_{t-1}] = I(p)$$
\end{lemma}

\begin{proof}
In every step, the advice is correct with probability $p$ and then the weight of $\tau$ is multiplied by $2p$, and faulty with probability $1-p$, and then the weight of $\tau$ is multiplied by $2(1-p)$. Hence:
$$E[\ell_t] = p(\log (2p) + \ell_{t-1}) + (1-p)(\log (2 - 2p) + \ell_{t-1}) = \ell_{t-1} + 1 - H(p)$$
as desired.
\end{proof}

Summarizing:

\begin{itemize}

\item $L_0 - \ell_0 = \log n$.

\item $L_t$ does not increase with $t$.

\item $\ell_t$ drifts upward at an expected rate of $I(p)$.

\item Informally, one would expect $\ell_t$ to overtake $L_t$ in $\frac{\log n}{I(p)}$ steps.

\item For fixed $p$ and sufficiently large $n$, the above estimate is very close to the truth, because the step size of $\ell_t$ is bounded (it is either $\log (2p)$ or $\log (2 - 2p)$), and the steps are independent. Hence Chernoff's bounds can be applied to show strong concentration around the expectation.

\end{itemize}
This concludes the proof of Theorem~\ref{thm:median}.
\end{proof}

\section{A $O(\sqrt{\Delta}\log n \cdot \log \log n)$ Query Algorithm}\label{sec:regular-upper}

This section is dedicated to proving Theorem \ref{thm:regular-upper}.
Algorithm $\algregular$ is described and analyzed. It performs almost optimally (up to lower order terms), assuming the noise parameter satisfies  $q<c/\sqrt{\Delta}$ for a sufficiently small positive constant $c$ (as opposed to  $q < (1- \eps)\Delta^{-1/2}$ as in Theorem   \ref{thm:weak-upper}). More precisely, in that regime, it finds the treasure in $\bigO(\algtwosepcost)$ queries in expectation.
We stress that, in contrast to Theorem \ref{thm:weak-upper}, in this section we do not allow the noise parameter to depend on the node.

Before we continue, let us note that taking a small enough $c$, the condition $q<c/\sqrt{\Delta}$ we are using here actually implies\footnote{Indeed, recall that for regular trees, Condition \condition (see Eq. \eqref{eq:condition})
reads
$q <\frac{1-\varepsilon- \Delta^{-1/4}}{\sqrt{\Delta} + \Delta^{1/4}}
$. Now, $\Delta \geq 2$ implies that $1-\Delta^{-1/4} \geq 1-2^{-1/4}$ and
$\Delta^{1/4} \leq \sqrt{\Delta}$. Hence
$\frac{1-\varepsilon- \Delta_v^{-1/4}}{\sqrt{\Delta} + \Delta^{1/4}} \geq \frac{1-2^{-1/4}-\varepsilon}{2} \F 1 {\sqrt{\Delta}}$.
We may set $\varepsilon = \frac{1-2^{-1/4}}{2}$
so that, as soon as $c < \frac{1-2^{-1/4}-\varepsilon}{2} = \frac{1-2^{-1/4}}{4}$, $q<c\Delta^{-1/2}$ implies Condition \condition with that choice of $\varepsilon$.} Condition \condition with $\varepsilon = (1-2^{-1/4})/2$.

Algorithm $\algregular$
runs two algorithms in parallel, namely,
$\fast$, and $\intermediate$.
Algorithm $\fast$ is actually  $\algsep$, except that it applies the local procedure with parameter $h$ being $\smallh=\lceil \kappa_2 \log \log n\rceil$ rather
than $\Theta(\log n)$. Algorithm $\intermediate$ is similar to $\algsep$, in the sense that  it uses $h$ being $\bigh = \lceil \kappa_1 \log n\rceil$. However it uses a different local exploration procedure, see more details in Section \ref{sec:mid}. $\kappa_1$ and $\kappa_2$ are constant independent of $n$ whose value will be determined later.
We will henceforth omit the ceiling $\lceil \cdot \rceil$ in the interest of readability.

Let us first recall some of the definitions that were introduced in Section \ref{sec:weak-upper}.
Here $T_h(u)$ denotes the tree of nodes at distance at most $h$ from $u$. Call a leaf $v \in T_h(u)$ a \emph{nominee} if its distance to $u$ is exactly $h$. Let  $\leaf[u]$ be the leaf on $T_h(u)$ closest to $\target$ if $\target \notin T_h(u)$ and $\leaf[u]= \tau$ otherwise.
Denote by $\mathcal{U}(u)$ the set of nominees that are not in the same component as $\leaf[u]$ in $T \setminus \{ u \}$.
Call a nominee $v$ {\em promising} if
$
\sum_{w\in [u,v \rangle} X_w\geq \F 2 3  h
$,
where  $X_w = 1$ if the advice at $w$ is pointing to $v$, $X_w = - 1$ if it is pointing to $u$, and
$X_w = 0$ otherwise. (Note that $X_u$ can never be $-1$.)
 Recall also that $u$ is called  \emph{$h$-misleading}, if one of the two following events happens:
 \begin{itemize}
 \item $\leaf[u] \neq \target$ and $\leaf[u]$ is not promising, or
 \item
there is some promising nominee in $\mathcal{U}(u)$.
 \end{itemize}
Let $\excellent$ be the event that no separator on the way to the treasure is $\smallh$-misleading.
The following claim is a direct consequence of Lemma \ref{lem:misleading} (regular tree case) and linearity of expectation, summing the query complexity of the $\lceil \log n\rceil$ separators on the way to the treasure.
\begin{claim}\label{claim:excellent}
$
\PP(\excellent) \cdot \queries\left(\fast \mid  \excellent \right) =
\bigO\left(\algtwosepcost\right).
$
\end{claim}
To bound the total expected number of queries, we run in parallel  algorithm $\intermediate$.
All that remains is then to prove that $\PP(\excellentc) \cdot \queries\left(\intermediate \mid  \excellentc \right) = \bigO\left(\algtwosepcost\right)$. In fact, we shall prove a stronger claim --- that the bound on the r.h.s is only $\bigO(\sqrt{\Delta}\log n)$. We further remark that following the arguments below and taking appropriate constant parameters and reducing the noise, one can get down to $\bigO(1+\sqrt{\Delta}/\log^k n)$ for any constant $k$. However, since proving a bound of $\bigO(\sqrt{\Delta}\log n)$ suffices for us, we shall stick to this bound.

\subsection{Algorithm $\intermediate$}\label{sec:mid}
As mentioned, $\intermediate$ is similar to $\algsep$ except that it uses a different local procedure. More precisely, recall that $\algsep$ executes Procedure $\local(u)$ by running $\algzplus$ on $T_h(u)$ until it either finds the treasure or finds a promising nominee,
and in the latter case, it declares that the treasure is on the connected component of $T \setminus \set{u}$ containing this nominee. In the context of Algorithm $\intermediate$, for technical commodity, we choose to run Procedure $\local(u)$ with a simpler  exploration routine which we call $\algy$. It is less efficient than $\algzplus$ but its simplicity will be useful for analyzing its behaviour in various, ``less clean'', circumstances. Indeed, we will need to analyze the performances of $\algy$, conditioning on the event $\excellentc$, implying that some parts of the tree have to be pointing in the wrong direction.
The fact that $\algy$ is less efficient than $\algzplus$ will not affect the final bound, as its running time will dominate the total running time with very low probability.

\paragraph{Algorithm $\algy$.} Recall in this section we only deal with $\Delta$-regular trees.
Define \emph{level $i$} as the set of all nodes at distance $i$ from the root.
At each round,  $\algy$ only compares nodes within a given level $i$.
Specifically, it goes to the node in level $i$ with most arrows pointing at it among the non-visited nodes in level $i$.
It only considers vertices whose parent has been explored already. The index $i$ is incremented modulo the depth of the tree $D$, on every round.
Below is a description in pseudocode.
The loop over $i$ explains the name $\algy$.

\RestyleAlgo{boxruled}
\LinesNumbered{}

\begin{algorithm}[ht]
  \caption{Algorithm $\algy$
}
\label{alg:simple}
Continuously loop over the levels $1, 2, \ldots, \diam$ \\
When considering level $i$, go to the yet unexplored reachable node at the current level (if one exists) that has most arrows pointing to it.
\end{algorithm}

In what follows we will analyse Algorithm $\algy$ conditioning on some parts of the tree being misleading. For readability considerations, the interested reader might wish to first see how it behaves on a simpler scenario, without any conditioning. The proof of this is was shown in a preliminary version of the paper (see \cite{advice1}).

\begin{lemma}\label{lem:algy}
Consider a (not necessarily complete) $\Delta$-ary tree. Then the expected number of queries of $\algy$ is
$\queries(\algy) = \bigO(\diam^3 \sqrt{\Delta})$.
\end{lemma}

In fact, a slightly more refined analysis shows that
$\queries(\algy) = \bigO(\diam^2 \sqrt{\Delta})$, but this is not needed for our current purposes, and so we omit it.

\subsection{Analysis of $\intermediate$ Conditioning On The Complement of $\excellent$}\label{sec:mainclaim}

To complete the proof of Theorem \ref{thm:regular-upper} we will show that
if $c$ small enough, then
\eq{
\PP(\excellentc) \cdot \queries\left(\intermediate \mid  \excellentc \right) = \bigO(\sqrt{\Delta} \log n).
}

\paragraph{Decomposing $\excellentc$.} At a high level, we seek to break $\excellentc$ into many elementary bad events.
Denote by $u_1, \ldots u_{\ell}$ the sequence of separators on the way to the treasure $\target$.
Note that $\ell \leq \lceil \log n \rceil$.
First,
\eq{
\excellentc = \bigcup_{i \leq \ell} \left\{ u_i \text{ is $h_2$-misleading}\right\}.
}
Using the union bound argument in Section \ref{app:app} (Claim \ref{claim:unionbound}),
\eql{
 \queries\left(\intermediate \bigcap  \excellentc \right) \leq \sum_{i \leq \ell} \queries\left(\intermediate \bigcap \text{$u_i$ is $\smallh$-misleading}\right),
}{eq:ubound}
where, to keep the equation light we write $\queries(\alg \bigcap X)$ in place of  $\queries(\alg\mid X)\cdot \PP(X)$ where $\alg$ is an algorithm and $X$ is an event.

Since we ultimately want to show that the left hand side in the previous equation is $\bigO(\sqrt{\Delta}\log n)$, it is sufficient to show that for every fixed $i\leq \ell$,
\eql{
 \queries\left(\intermediate\bigcap \text{$u_i$ is $\smallh$-misleading}\right) =
 \bigO(\sqrt\Delta).
}{eq:sepbound}
From now on, we fix $i$ and focus on the case where $u_i$ is $\smallh$-misleading.
Recall that algorithm $\intermediate$, just as $\algsep$, proceeds in phases of local exploration,
running also an exhaustive search in parallel to handle the case that one of the local explorations ends with a wrong answer.
Denote by
$\good$ the event that all separators on the way to the treasure, namely, $u_1, \ldots, u_\ell$, are not $\bigh$-misleading.
Under $\good$, each local explorations outputs the correct output (the next separator leading to the treasure) and the local exploration phases amount to running $\algy$ on $T_{\bigh}(u_j)$ for $j \leq \ell$. Now,
\eq{
\queries\left(\intermediate\bigcap  \text{$u_i$ is $h_2$-misleading} \right)
 & =
\queries\B{\intermediate \bigcap  \B{\text{$u_i$ is $h_2$-misleading}\cap \good}}
\\ & +
\queries\left(\intermediate\bigcap \left( \text{$u_i$ is $h_2$-misleading}\cap \neg \good \right) \right).
}
By Lemma \ref{lem:misleading} (regular tree case),
\[
\PP(\neg\good) \leq 2(1 - \varepsilon)^{h_1}
=
2(1 - \varepsilon)^{\kappa_1 \log n}.
\]
Recall that Condition \condition is satisfied with the constant $\varepsilon = (1-2^{-1/4})/2$,
and so taking $\kappa_1$ to be a large enough constant, gives that $\PP(\neg\good) < 1/n$. This means that if $\good$ does not hold, it is fine to resort to exhaustive search, as the second term above becomes $\bigO(1)$.
Note that given that $u_i$ is $h_2$-misleading may affect the advice which is relevant for the local explorations corresponding to other separators. Hence, we have:
\eq{
& \queries\left(\intermediate\bigcap  \text{$u_i$ is $h_2$-misleading} \right)
\\ & \leq
\sum_{j \leq \log n}
 \queries\left(\algy\left(T_{\bigh}(u_j)\right) \bigcap \left(\text{$u_i$ is $h_2$-misleading} \cap \good\right) \right) + \bigO(1)
\\ & \leq
\sum_{j \leq \log n}
 \queries\left(\algy\left(T_{\bigh}(u_j)\right) \bigcap \text{$u_i$ is $h_2$-misleading} \right)
 +\bigO(1)
.
}
The first inequality is by linearity of expectation,
and the last inequality follows from the fact that for every algorithm $A$ and every two events $E_1 \subseteq E_2$,
$\queries(A \bigcap E_1) \leq \queries(A \bigcap E_2)$.

For the sake of lightening notations, we henceforth refer to $u_j$ as $\newsource$ and $u_i$ as $u$. This choice of notations reflects the fact that
we are rooting the tree at $u_j = \source'$ and running $\algy$ on $T_{h_1}(\newsource)$.
The fact that $\newsource$ and $u$ are separators is not relevant in this analysis.
We also denote by $\newtarget$ the leaf on $T_{h_2}(u)$ that is closest to $\treasure$ and by $\target'$ the leaf of $T_{h_1}(u)$ that is closest to $\target$ or simply $\target$ if
$\target \in T_{h_1}(u)$.
With these notations Equation \eqref{eq:sepbound} immediately follows once we prove:

\begin{lemma}\label{lem:corelem2}
For every $\sigma', u \in T$,
\eq{
 \queries\left(\algy\left(T_{\bigh}(\sigma')\right) \bigcap \text{$u$ is $h_2$-misleading}\right) =\bigO\left( \frac{\sqrt{\Delta}}{ \log n}\right).
}
\end{lemma}

\paragraph{Decomposing the event $\{\text{$u$ is $h_2$-misleading}\}$.}
So far we saw that it is sufficient to analyse the events  where one separator is $h_2$-misleading.
We now pursue decomposing these events into even smaller ones.
To this aim the following definition is convenient.
\begin{definition}\label{def:badevents}
Let $a,b \in T$ be two nodes such that $a$ is the closest one to $\tau$ out of the nodes in
$\pathcc{a}{b}$. Noting that a vertex can never point to itself:
\begin{itemize}
\item
For $S \subseteq \pathoc{a}{b}$, denote by $\badsides{a,b}$ the event that none of the nodes of $S$ point towards $a$ and none towards $b$.
\item
For $S \subseteq \pathco{a}{b}$, denote by $\badup{a,b}$ the event that the nodes of $S$ all point towards $b$.
\end{itemize}
See Figure \ref{fig:cases}.
\end{definition}
\begin{claim}\label{claim:estimates}
For every $a,b$ and $S$ as in Definition \ref{def:badevents},
\begin{itemize}
\item
$\PP\left(  \badsides{a,b} \right)
\leq q^{\lvert S \rvert}$,
\item
$\PP\left(\badup{a,b}\right)
 \leq \left( \F q {\Delta}\right)^{\lvert S\rvert}$.
\end{itemize}
\end{claim}

Let us now see in more detail what it means for a node $u$ to be $\smallh$-misleading.
Several cases need to be considered.
\begin{enumerate}
\item
$\leaf[u] \neq \target$ and $\tau_u$ is not promising. In this case $\lvert \path{u}{\leaf[u]} \rvert = \smallh$ and  the sum of advice on $\pathco{u}{\tau_u}$ is strictly less than $\F23 h_2$. In this case, at least one of the following two must be true:
\begin{enumerate}
\item
There are
$\kside \smallh$ locations on the path $\pathco{u}{\leaf[u]}$ where the advice points outside of the path
(the value of the corresponding $X_i$'s is $0$).
This corresponds to
$\badsides{\newtarget, u}$ for some set $S \subseteq\pathco{u}{\leaf[u]}$ of size
$|S| = \kside \smallh$.
\item
There are $\kup \smallh$ locations on $\pathoo{u}{\tau_u}$ that point towards $u$ (the value of the corresponding $X_i$'s is $1$). This corresponds to $\badup{\leaf[u], u}$ for some set $S \subseteq\pathcc{u}{\leaf[u]}$ of size $\lvert S \rvert = \kup \smallh$.
\end{enumerate}
To see why, let $X^{\tau_u}$ be the number of pointers pointing to $\tau_u$, $X^u$ the number of pointers pointing to $u$ and $X^{0}$ the number pointers pointing to the sides. We have $X^{\tau_u}-X^u\leq 2h_2/3$, and hence $X^u-X^{\tau_u}\geq -2h_2/3$. If $(a)$ does not hold, then $X^{0}<h_2/6$ which implies that
$X^{\tau_u}+X^u\geq 5h_2/6$. Summing the two equations, we get $X^u\geq h_2/12$, as stated in (b).
\item Some  $v\in \mathcal{U}(u)$ is promising. In this case there must be some $\F 2 3\smallh$ locations on $\path{u}{v}$ that point towards $v$. This corresponds to $\badup{u,v}$ for some
$S \subseteq\badup{\pathcc{v}{u}}$ of size $|S|=\F 2 3 \smallh$.
\end{enumerate}
Define
\begin{itemize}
    \item
$\mathcal{C}(u)=\{S\subseteq\path{u}{\leaf[u]}\mid |S|=\kside \smallh\}$,
\item $\mathcal{D}(u) = \{S\subseteq\path{u}{\leaf[u]}\mid |S|=\kup \smallh\}$, and
\item $\mathcal{E}(u)=\{(v,S)\mid \mbox{$v \in \mathcal{U}(u)$,  } S\subseteq\pathcc{u}{v}, \mbox{ and } |S|=\F 2 3 \smallh\}$.
\end{itemize}
Combining Definition \ref{def:badevents} with the previous paragraph, yields
\eql{
\{ \text{$u$ is $\smallh$-misleading} \}
& \subseteq  \left\{\leaf[u] \neq \target \text{ and }\newtarget \text{ is not promising} \right\} \bigcup  \bigcup_{v \in \mathcal{U}(u)}\left\{ v \text{ is promising} \right\}
\\ & \subseteq
\bigcup_{S \in \mathcal{C}(u)} \badsides{\newtarget, u}
\bigcup_{S \in \mathcal{D}(u)} \badup{\newtarget, u}
\bigcup_{(v,S) \in \mathcal{E}(u)} \badup{u,v}.
}{eq:cup}
In fact,
$\mathcal{E}(u)$ needs to be further decomposed taking into account the position of $u$ within the tree rooted at $\sigma'$ and the path from $\sigma'$ to $\tau'$, see Figure \ref{fig:cases}. For each $(v,S)\in\mathcal{E}(u)$, let \[k(v)= \lvert\pathcc{u}{v} \cap \pathcc{\sigma'}{\tau'}\rvert.\] For each non-negative integer $k\geq 0$, let
\[
\mathcal{E}_k(u)=\{(v,S)\in \mathcal{E}(u)\mid k(v)=k\}.
\]
Clearly, as $\lvert\pathcc{u}{v}\rvert\leq h_2$, we have $\mathcal{E}(u)=\cup_{k=0}^{\smallh}\mathcal{E}_k(u)$.\\
\begin{figure}[!ht]
    \centering
    \includegraphics[width=0.9\textwidth]{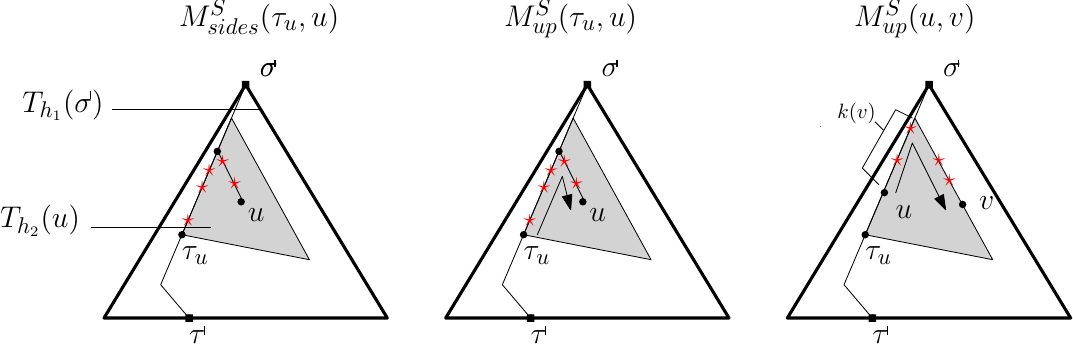}
        \caption{\small Different relative positions of $u,\newtarget$ and $\sigma'$. The path $\pathcc{u}{\newtarget}$ and different mistake patterns. On the left, mistakes (depicted as red stars) point outside of $\path{u}{\leaf[u]}$, on the middle tree they point towards $u$ and on the right one, they point towards a nominee of $T_{\smallh}(u)$, $v \in \mathcal{U}(u)$.}\label{fig:cases}
\end{figure}

Using the union bound (Claim \ref{claim:unionbound}) as in Eq. \eqref{eq:ubound}, the aforementioned decomposition in Eq. \eqref{eq:cup} implies:
\eql{
 \queries\left(\algy\left(T_{\bigh}(\newsource)\right)  \bigcap \text{$u$ is $\smallh$-misleading}\right)
 &\leq  \sum_{S \in \mathcal{C}(u)}
  \queries\left(\algy(T_{\bigh}(\newsource))\bigcap\badsides{u,\newtarget}\right)\nonumber\\
&+\sum_{S \in \mathcal{D}(u)} \queries\left(\algy(T_{\bigh}(\newsource))\bigcap  \badup{u,\newtarget}\right)\nonumber\\
&+\sum_{k=0}^{\smallh}\sum_{(v,S) \in \mathcal{E}_k(u)} \queries\left(\algy(T_{\bigh}(\newsource))\bigcap \badup{u,v}\right)
}{eq:maineq}
To prove Lemma \ref{lem:corelem2}, our goal will be to show that each sum in the above equation is at most $\bigO(\sqrt{\Delta} /  \log n)$.

\subsection{Analysing Atomic Expressions}\label{sec:corelem2}

To prove that each sum is indeed $\bigO(\sqrt\Delta / \log n)$ we use the following two lemmas (proved in Section \ref{sec:resilience}), which encapsulate the core of this proof, namely, the resilience of $\algy$ to certain kinds of error patterns.
\rLemma{lemResilienceE}{
Consider a tree $T$ rooted at $\sigma$ with treasure located at $\target$. Let $a,b\in T$ be two nodes such that $a$ is the closest one to $\tau$ out of the nodes in $[a,b]$. Then, for any $S \subseteq \pathoc{a}{b}$, we have
\eq{
\queries\left(\algy \mid \badsides{a,b}\right)
= \bigO\left( \diam^4 \Delta^{\F{|S| + 1}2}\right).
}
}
\rLemma{lemResilienceD}{
Consider a tree $T$ rooted at $\sigma$ with treasure located at $\target$. Let $a,b\in T$ be two nodes such that $a$ is the closest one to $\tau$ out of the nodes in $[a,b]$. Then, for any $S \subseteq \pathoc{a}{b}$, we have
\eq{
\queries\left(\algy \mid  \badup{a,b}\right)
= \bigO\B{\diam^4 \Delta^{K + \R2} 3^{|S|}},
}
where $K = |S \cap [\sigma, \tau]|$.
}
As a first step to bounding the three sums of Eq. \eqref{eq:maineq}, note that:
\begin{align}
\lvert \mathcal{C}(u)\rvert &\leq 2^{\smallh}\label{s1} \\
\lvert \mathcal{D}(u)\rvert &\leq 2^{\smallh}, \label{s2} \\
\lvert \mathcal{E}_k(u)\rvert &\leq 2^{\smallh}  \Delta^{\smallh - k}.\label{s3}
\end{align}
Indeed, $\mathcal{C}(u),\mathcal{D}(u)$ are subsets of a path of length $\smallh$. For the last term, the number of $v \in \mathcal{U}(u)$ at distance $\smallh$ from $u$ for which $k(v) = k$ is bounded by $\Delta^{\smallh-k}$.

We are now ready to bound the three sums.
\paragraph{Bounding the first term in Eq.~\eqref{eq:maineq}.}
We consider $S \in \mathcal{C}(u)$,
so
$S\subseteq \path{u}{\leaf[u]}$ and $|S|=\kside \smallh$, and $\leaf[u]$ is the closest to $\tau$ of all the nodes on the path.
By Lemma \ref{lemResilienceE},
\[
\queries\left(\algy(T_{\bigh}(\newsource))\mid\badsides{\newtarget, u}\right)
=
\bigO\B{\bigh^4 \Delta^{\F{|S| + 1}2 }}.
\]
According to Claim \ref{claim:estimates},
\[
\PP(\badsides{\leaf[u], u})\leq q^{|S|}.
\]
Combining these bounds and \eqref{s1} yields
\begin{align*}
\sum_{S \in \mathcal{C}(u)}
\queries\left(\algy(T_{\bigh}(\newsource))\bigcap\badsides{\newtarget, u}\right)
&=
\bigO\left( 2^{\smallh} \cdot q^{|S|} \cdot \bigh^4  \Delta^{\F{|S| + 1}2}
\right)\\
&=
\bigO\B{ \sqrt\Delta \cdot 2^{\smallh} \cdot c^{|S|} \cdot \bigh^4}
,
\end{align*}
because $q<c/\sqrt{\Delta}$.
Recall that $\bigh = \kappa_1\log n$, $\smallh = \kappa_2 \log \log n$, and $|S| = \R6 h_2$.
$\kappa_1$ was already set to be some constant.
For any constant $\kappa_2$, taking a small enough constant $c$ guarantees that
the previous expression is $\bigO(\sqrt{\Delta} / \log n)$ as needed.
\paragraph{Bounding the second term in Eq.~\eqref{eq:maineq}.}
$S \in \mathcal{D}(u)$,
so
$S\subseteq \path{u}{\leaf[u]}$ and $|S|=\kup \smallh$.
Therefore, by Lemma \ref{lemResilienceD}, and noticing that $K \leq |S|$ and $3^{|S|} \leq 2^{h_2}$, we have
\[
\queries\left(\algy(T_{\bigh}(\newsource))\mid  \badup{\newtarget, u}\right)
=
\bigO\B{\bigh^4 \Delta^{|S| + \R2} 2^{h_2}}.
\]
Combined with Claim \ref{claim:estimates} and \eqref{s2}:
\eq{
\sum_{S \in \mathcal{D}(u)}
\queries\left(\algy(T_{\bigh}(\newsource))\bigcap  \badup{\newtarget,u}\right)
& =
\bigO\left( 2^{\smallh} \cdot \left( \frac{q}{\Delta}\right)^{|S|} \cdot \bigh^4  \Delta^{|S| + \R{2}} 2^{\smallh} \right)
\\ & =
\bigO\left( \sqrt\Delta \cdot 4^{\smallh} \cdot q^{|S|} h_1^4 \right).
}
Again, since $|S| = \R{12} h_2$, then for any constants $\kappa_1,\kappa_2$, the constant $c$ can be chosen so that this is $\bigO(\sqrt{\Delta} / \log n)$.
\paragraph{Bounding the third term in Eq.~\eqref{eq:maineq}.}

$(v,S) \in \mathcal{E}_k(u)$,
where $v \in \mathcal{U}(u)$, $S\subseteq \pathcc{u}{v}$, and  $|S|=\F 2 3 \smallh$.
Also,
$|\pathcc{u}{v} \cap \pathcc{\sigma'}{\tau'}| = k$, and
so $|S \cap \pathcc{\sigma'}{\tau'}| \leq k$.
As $v \in \mathcal{U}(u)$, then $u$ is the closest to treasure of the vertices on $\pathcc{u}{v}$.
By Lemma \ref{lemResilienceD},
\[
\queries\left(\algy(T_{\bigh}(\newsource))\mid \badup{u,v}\right) =
\bigO\B{h_1^4 \Delta^{k + \R{2}} 3^{h_2} }
\]
Combined with \eqref{s3} and Claim \ref{claim:estimates}:
\eq{
&\sum_{k=0}^{\smallh}\sum_{(v,S) \in \mathcal{E}_k(u)} \queries\left(\algy(T_{\bigh}(\newsource))\bigcap \badup{u,v}\right)
\\ & =
\bigO\left(\sum_{k \leq \smallh} 2^{\smallh} \Delta^{\smallh-k} \cdot  \left( \frac{q}{\Delta}\right)^{\F 2 3 \smallh}  \bigh^4 \cdot \Delta^{k + \R2} 3^{\smallh}\right)
\\ & =
\bigO\left( \sqrt\Delta \cdot \smallh 6^{\smallh} \bigh^4 \B{q^2 \Delta}^{\F 1 3 \smallh} \right).
\\ & =
\bigO\left( \sqrt\Delta \cdot \smallh 6^{\smallh} \bigh^4 \cdot c^{\F 1 3 \smallh} \right).
}
Similarly to the two previous sums, this whole expression can be made as small as $\bigO(\sqrt{\Delta} / \log n)$.

Note that we assumed for simplicity that $u$, $\tau_u$ and $v$ are all inside $T_{h_1}(\sigma')$. If they are not, we take nodes that are the closest to them on this subtree, which can only improve the bounds.

This concludes the proof of Lemma \ref{lem:corelem2} and hence completes the proof of Theorem \ref{thm:regular-upper}.

\subsection{The Lemmas About the Resilience of $\algy$}\label{sec:resilience}

Recall that Algorithm $\algy$ loops over the levels $1, 2, \ldots, \diam$, and
when considering level $i$, it goes to the yet unexplored reachable node at the current level (if one exists) that has most arrows pointing to it on the path from the root. In this section, we prove two lemmas that bound the expected number of queries done by Algorithm $\algy$ conditioning on some ``bad'' events.

Before stating the lemmas we need a couple of definitions that will be used in both corresponding proofs. We say that a node $v$ is a {\em competitor} of $u$ if it has the same depth as $u$. A competitor $v$ of $u$  {\em beats}  $u$ if the number of pointers pointing to it on the path from the root is at least  the number pointing to $u$.

\begin{figure}[!ht]
    \centering
    \includegraphics[width=0.6\textwidth]{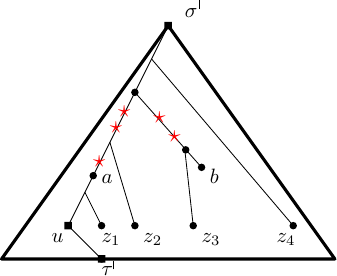}
        \caption{\small
        Notations introduced in the proof of Lemma \ref{lemResilienceE} and \ref{lemResilienceD}.
        Points of $S$ are depicted in red. On the figure $n(z_1)=0, m(z_1)=0$, $n(z_2)=1, m(z_2)=0$, $n(z_3)=3,m(z_3)=2$ and $n(z_4)=3,m(z_4)=0$.}\label{fig:resilience3}
\end{figure}

\lemResilienceE
\begin{proof}
As in the proof of Lemma \ref{lem:algy}, we break the number of queries made by $\algy$ conditioning on $\badsides{a,b}$ into a sum of random variables $Q_j$ which correspond to the expected number of queries needed to discover the $j$-th node $u_j$ on the path $\pathcc{\source}{\target}$ once the $(j-1)$-th node, $u_{j-1}$, was discovered.
Each $Q_j$ is bounded above by $D$ times the expected number of competitors of $u_j$ that beat it.
This is because each loop takes $D$ steps, and only a subset of the competitors that beat $u_j$ will actually be checked by $\algy$ on level $j$ before trying the correct node.
Hence, an upper bound on the expected number of competitors who beat any given $u\in \pathcc{\source}{\target}$ translates to an upper bound on $\queries\left(\algy(T) \mid \badsides{\target,\source}\right)$ by multiplying it by $D^2$.

Let $u$ be any node on the path $\pathcc{\source}{\target}$, and $z$ be a competitor of $u$.
Define $k(z)$ as half the distance between $z$ and $u$, namely, $d(z,u)/2$, and denote $n(z) := \lvert S \cap \pathcc{\sigma}{\tau} \cap \pathcc{u}{z} \rvert$ and $m(z):=\lvert S \cap \pathcc{\sigma}{\tau}^c \cap \pathcc{u}{z} \rvert$.
See Figure \ref{fig:resilience3} for illustraion.

First note, that since all advice of $S \subseteq [a,b]$ points sideways w.r.t.\ to this path, then any of it which is on the path $[u, z]$ also points sideways w.r.t.\ it, except possibly at one point, the least common ancestor between $u$ and $z$, which may actually point towards $z$. The different cases are seen in Figure \ref{fig:resilience3}:
\begin{itemize}
\I
For $z=z_1$, the paths do not intersect at all.
\I
In the case that $z= z_2$, if the least common ancestor of $u$ and $z_2$ is a member of $S$, then it could point towards $z_2$, and that would be sideways w.r.t.\ $[a,b]$.
\I
For $z=z_3$, the least common ancestor of $b$ and $z_3$ could point towards $z_3$.
\I
For $z=z_4$, the least common ancestor of $a$ and $b$ could point towards $z_4$.
\I
There is also the case where $a \notin \path{\sigma}{\tau}$, which is not depicted on Figure \ref{fig:resilience3}. The analysis remains valid, and in fact $n(z)=0$ for all competitors $z$.
\end{itemize}
This one special vertex, if it exists,
conditioned on that it points sideways w.r.t.\ $[a,b]$, points towards $z$ with probabilty $1/(\Delta - 2)$, and otherwise points sideways w.r.t.\ $[u,z]$.

Fix $k, n$ and $m$, and consider a competitor $z$ such that $k(z) = k$, $n(z) = n$, and $m(z) = m$.
On the path $\pathcc{u}{z}$ the number of advice remaining to be sampled is $2k - n - m - 1$. For any of these samples $s$, let $X_s$ be 1 if the pointer points to $u$, $-1$ if it points to $z$, and 0 otherwise.
By Lemma \ref{lem:minortec}:
\eq{
\PP\B{z \text{ beats } u}
& \leq
\B{1 - \R{\Delta -2}}
\PP\B{\sum_{s=1}^{2k-1-n-m} X_s \leq 0}
+
\R{\Delta - 2}
\PP\B{\sum_{s=1}^{2k-1-n-m} X_s \leq 1}
\\ & \leq
\BF1{\sqrt\Delta}^{2k-1-n-m}
+
\F3{\Delta - 2}
\BF1{\sqrt\Delta}^{2k-2-n-m}
\\ & =
\B{1 + \F{3\sqrt{\Delta}}{\Delta - 2}}
\BF1{\sqrt\Delta}^{2k-1-n-m}
\leq
7 \cdot \BF1{\sqrt\Delta}^{2k-1-n-m},
}
as $\Delta \geq 3$.
For fixed $k,n,m$ there are at most $\Delta^{k-m}$ nodes $z$ with $k(z)=k$ and $m(z)=m$. Also, for each such node, $n + m \leq 2k$. Hence, the total expected number of competitors that beat $u$ is at most:
\eq{
\sum_{k \leq D, n+m\leq 2k}  \Delta^{k-m} \cdot 7 \BF1{\sqrt\Delta}^{2k-1-m-n}=\bigO\left(\sum_{k \leq D, n+m\leq 2k} \Delta^{(n +1 -m)/2}\right).
}
For each choice of $k$ there is exactly one corresponding value of $n$. This $n$ satisfies $n \leq |S|$. There are also at most $D$ choices for $m$. Thus, the above is at most
\eq{
  \bigO\left(D^2 \Delta^{(|S|+1)/2}\right).
}
Multiplying this bound by $D^2$ gives the desired upper bound on $\queries\left(\algy(T) \mid \badsides{\target,\source}\right)$.
\end{proof}

\lemResilienceD
\begin{proof}
The general structure of the proof is similar to the proof of Lemma \ref{lemResilienceE}. Let $u$ be a node on the path $\pathcc{\source}{\target}$. Our aim is to show that the expected number of competitors of $u$ that beat it is $\bigO(D^2  \Delta^{K+\R{2}} 3^{|S|})$. Once this is established, multiplying this bound by $D^2$ gives the desired bound on the number of queries.

As in the proof of Lemma \ref{lemResilienceE},
let $z$ be a competitor of $u$.
Define $k(z)$ as half the distance between $z$ and $u$, namely
$k(z) := d(z,u)/2$. Denote $n(z) := \lvert S \cap \pathcc{\sigma}{\tau} \cap \pathcc{u}{z} \rvert$, and $m(z):=\lvert S \cap \pathcc{\sigma}{\tau}^c \cap \pathcc{u}{z} \rvert$.

Fixing $k,n$ and $m$, take a competitor $z$ such that $k(z) = k$, $n(z) = n$, and $k(z) = k$. For any of the nodes $s$ on the path from $u$ to $z$, let $X_s$ be 1 if the pointer points to $u$, $-1$ if it points to $z$, and 0 otherwise.
By Lemma \ref{lem:minortec}, the probability that such a $z$ beats $u$ is:
\[
\PP\B{z \text{ beats } u}\leq
 \PP\left( \sum_{s=1}^{2k-1-n-m} X_s \leq n+m\right)
\leq
3^{n+m} \Delta^{n+m-k +\R2}.\]
There are at most $\Delta^{k-m}$ such nodes $z$.
We bound the probability that each of these nodes $z$ beats $u$ using the trivial bound $1$ or the one above, depending on whether $n+m\leq k$ or $n+m> k$.
Hence the total expected number of competitors of $u$ who beat it is at most
\eq{
\sum_{k \leq \diam, n+m\leq k } \Delta^{k-m} \cdot 3^{n+m} \Delta^{n+m-k +\R2}
+ \sum_{k \leq D,n +m > k} \Delta^{k-m}.
}
Since $n + m \leq |S|$, and $n \leq K$, the term on the left is at most:
\[
3^{|S|} \sum_{k \leq D,n+m \leq k} \Delta^{K+\R{2}}
\leq
3^{|S|} \cdot D^2 \cdot \Delta^{K + \R{2}},
\]
where we used the fact that there are most $D$ distinct values for $k$ and $D$ distinct values for $m$, while there is only one choice of $n$ for each $k$.
As for the second term, since $n+m>k$, then it is at most:
\[
\sum_{k \leq D,n + m > k} \Delta^n
\leq
\sum_{k \leq D,n + m > k} \Delta^K
\leq
\sum_{k, m \leq D} \Delta^K
\leq
D^2 \cdot \Delta^K,
\]
concluding the proof.
\end{proof}

\section{Complementary Proofs}\label{app:app}

\subsection{Another Large Deviation Estimate}\label{sec:largb}
Here, we introduce another large deviation estimate used for the analysis of the query algorithm with uniform noise. It gives better results for large $h$, yet works only for identical random variables, and so suits regular trees, unlike Lemma \ref{lem:maintec}. It is used in the proof of the query algorithm presented in Theorem \ref{thm:regular-upper}, specifically in Lemmas \ref{lem:algy}, \ref{lemResilienceD} and \ref{lemResilienceE}.

\begin{lemma}\label{lem:minortec}
Consider random variables $X_i$ taking values $\{-1, 0, 1\}$ with respective probabilities
$(\frac{q}{\Delta}, q \left(1- \frac{2}{\Delta}\right), 1- q+\frac{q}{\Delta})$.
If $q< \F{c}{\sqrt\Delta}$, where $c<1/9$, then for all $0 \leq h \leq l$,
\[
\PP\left( \sum_{i=1}^{\ell} X_i \leq h\right)
\leq (3\sqrt{\Delta})^h \Delta^{-\ell/2}.
\]
\end{lemma}
\begin{proof}
Denote by $N,Z,P$ the number of $X_i$'s respectively equal to $-1,0,1$. We thus have
\begin{equation}\label{eq:IJ}
    N+Z+P=\ell.
\end{equation}
Under the assumption that $\sum_{i=1}^{\ell} X_i \leq h$, we also have
\begin{equation}\label{eq:-IJ}
    -N+P \leq h.
\end{equation}
These two equations characterize all possible values $N,Z,P$ can take. We denote by $A$ the set of such tuples.
Decomposing according to the $N$ locations of $-1$'s, $P$ locations of $1$'s and $Z$ locations of $0$'s we obtain
\eq{
\PP\left( \sum_{i=1}^{\ell} X_i \leq h\right)
&\leq
\sum_{N,Z,P \in A} \left(1-q+\frac{q}{\Delta}\right)^{P} q^{Z} \left(\frac{q}{\Delta}\right)^{N}
\binom{\ell}{N} \binom{\ell - N}{P}.
}
We crudely bound $\left(1-q+\frac{q}{\Delta}\right)^{P}\leq 1$ and obtain that
the aforementioned bound is at most
\eq{
\sum_{N,Z,P \in A} q^{Z} \left(\frac{q}{\Delta}\right)^{N}
\binom{\ell}{N} \binom{\ell - N}{P}.
}
Subtracting Eq. \eqref{eq:-IJ} from Eq. \eqref{eq:IJ}, we obtain
\eq{2N +Z \geq \ell  - h.}
Hence, for fixed $N,Z,P \in A$, we bound one of the inner terms:
\[
q^{Z} \left(\frac{q}{\Delta}\right)^{N}\leq \hat{q}^{Z+2N}\leq \hat{q}^{\ell  - h},\] where $\hat{q}:=\max\{q, (\frac{q}{\Delta})^{1/2}\}$. Under the assumption that $q<\frac{c}{\sqrt{\Delta}}$, we have $\hat{q}\leq \frac{\sqrt{c}}{\sqrt{\Delta}}$.
Hence, altogether
\[
\PP\left( \sum_{i=1}^{\ell} X_i \leq h\right)
 \leq \left(\frac{\sqrt{c}}{\sqrt{\Delta}}\right)^{\ell  - h} \sum_{N,Z,P \in A}
\binom{\ell}{N} \binom{\ell - N}{P}.
\]
There are at most $3^\ell$ possible outcomes for the sequence $X_i$ because it is comprised of $\ell$ variables that can take $3$ values each. Hence the term $\sum_{N,Z,P \in A}\binom{\ell}{N} \binom{\ell - N}{P}\leq 3^{\ell}$.
Altogether
\[
\PP\left( \sum_{i=1}^{\ell} X_i \leq h\right)
\leq 3^{\ell}  \cdot  \left(\frac{\sqrt{c}}{\sqrt{\Delta}}\right)^{\ell-h}.
\]
Finally, as $\sqrt{c}<1/3$,
the right hand side is thus
\eq{\leq (3\sqrt{\Delta})^h \Delta^{-\ell/2}.}
as stated.
\end{proof}

\subsection{Algorithm $\algy$ without Conditioning}\label{app:y}
\begin{lemma*}[Lemma \ref{lem:algy} restated]
Consider a (not necessarily complete) $\Delta$-ary tree. Then
$\queries(\algy) = \bigO(\diam^3 \sqrt{\Delta})$.
\end{lemma*}
\begin{proof}
Denote by $\NL(u)$ the number of nodes on the same depth as $u$ which have more discovered arrows than $u$ pointing to them.
This definition is central because of the following observation. The number of moves needed before finding $u_{i+1}$ once $u_i$ has been found is less than $\bigO(\diam \NL(u_i))$. Indeed, once $u_i$ is discovered, only \emph{a subset of} the nodes which have more arrows pointing to them than $u_{i+1}$ on layer $i+1$ are tried before $u_{i+1}$ (at step $(2)$ in the pseudocode description).
The loop over the levels (at step $(1)$) induces a multiplicative factor of $\bigO(\diam)$.

Using linearity of expectation, it only remains to estimate
$\EE\left(\NL(u_i)\right)$ where $u_i$ is the ancestor of the treasure at depth $i \leq d$.
There are at most $ \Delta^{\ell}$ nodes on layer $i$ at distance $2\ell - 1$ from $u_i$, for every $1\leq \ell\leq i$.
Moreover, for each of these nodes, the probability that it has at least as many arrows pointing towards it than $u_i$ exactly corresponds to $\prob{\sum_{j = 1}^{2\ell-1} X_j \leq 0}$, with the notations of Lemma \ref{lem:minortec}.

Indeed, when comparing the amount of advice pointing to two different nodes $u$ and $v$,
only the nodes of $\pathoo{u}{v}$ matter.

When estimating the probability that $v$ beats $u$,
each random variable $X_j$  has to be interpreted as taking value $-1$ if the advice points towards $v$, $0$ if it points neither to $u$ nor $v$, and $+1$ if it points towards $u$.
In the case that $u=u_{j}$ and $v$ is another node on layer $j$,
these events happen respectively with probability
$q/\Delta$,
$q(1-2\F 1 \Delta)$, and
$1-q + q/\Delta$.

This means that for each $i$,
\eq{
\EE\left(\NL(u_i)\right)\leq \sum_{\ell=1}^{i}
\PP\left(\sum_{j= 1}^{2\ell-1} X_j \leq 0\right)
 \Delta^\ell\leq \sum_{\ell=1}^{d}\PP\left(\sum_{j = 1}^{2\ell-1} X_j \leq 0\right)\Delta^\ell.}
By Lemma \ref{lem:minortec} this is at most
\[
\bigO\B{\sum_{\ell=1}^{d} \Delta^{-\ell + \R{2}} \cdot \Delta^\ell}
=
\bigO(D\sqrt{\Delta}) = \bigO\left(D\sqrt{\Delta}\right).
\]
\end{proof}

\subsection{Special Form of Union Bound}\label{sec:unionbound}
\begin{claim}\label{claim:unionbound}
Let $A$ be an event that can be decomposed as the union of events $(A_i)_{i \in I}$, $A \subseteq \bigcup_{i \in I} A_i$. Let $X$ be a random variable.
\eq{
\EE(X \mid A) \PP(A) \leq \sum_i \EE(X \mid A_i) \PP(A_i)
}
\end{claim}
\begin{proof}
We denote by $\chi (B)$ the indicator function of event $B$. Then
\begin{align*}
\EE(X \mid A) \PP(A)
& = \EE(X \cdot \chi(A))
\leq \EE\left(X \cdot \chi\left(\bigcup_i A_i\right)\right)
\leq \EE\left(X \cdot \sum_i \chi(A_i)\right)
\\ & = \sum_i  \EE\left(X \cdot \chi(A_i)\right)
=  \sum_i \EE(X \mid A_i) \PP(A_i),
\end{align*}
where we used the union bound in the form $\chi(\bigcup_i A_i) \leq \sum_i \chi A_i$ and then linearity of expectation.
\end{proof}

\end{document}